%% file: newftspanners.tex
\documentclass[11pt]{article}

\usepackage{amsmath}
\usepackage{amsfonts}
\usepackage{amsthm}
\usepackage{amssymb}
\usepackage[utf8]{inputenc}

\usepackage{textcomp}
\usepackage{cleveref}

\usepackage[margin=1in]{geometry}
\usepackage{indentfirst}

\usepackage{authblk}
\usepackage{abstract}
\usepackage{pbox}
\usepackage{comment}

\usepackage[linesnumbered, boxruled]{algorithm2e}
\usepackage{pgfplots}
\pgfplotsset{compat=1.4}
\usetikzlibrary{shapes}
\usetikzlibrary{plotmarks}

\usepackage{framed}
\usepackage{mdframed}
\usepackage{placeins}
\usepackage{subcaption}
\usepackage{afterpage}
\usepackage{hyperref}
\usepackage{courier}
\usepackage{algorithm2e}

\newcommand{\bee}{\mathcal{B}}
\newcommand{\dee}{\mathcal{D}}

\DeclareMathOperator{\dist}{dist}

\newtheorem{theorem}{Theorem}

\newtheorem{lemma}{Lemma}
\newtheorem{definition}{Definition}
\newtheorem{conjecture}{Conjecture}

\newtheorem*{definition*}{Definition}
\newtheorem*{theorem*}{Theorem}

\newtheorem*{lemma*}{Lemma}

\title{Optimal Vertex Fault Tolerant Spanners\\ (for fixed stretch)}
\author[1]{Greg Bodwin}
\author[2]{Michael Dinitz}
\author[1]{Merav Parter}
\author[1]{Virginia Vassilevska Williams}
\affil[1]{MIT CSAIL}
\affil[2]{Johns Hopkins University}
\date{}
\begin{document}
\maketitle
\thispagestyle{empty}
\input{abstract}
\pagebreak
\pagenumbering{arabic}

\section{Introduction}
\input{intro}

\input{technical}

\section{Upper Bound for $3$-Spanners}

First, we have:
\begin{lemma} \label{lem:3-reg}
Suppose that Theorem \ref{thm:upper-bounds} holds for $k=2$ for all graphs $G$ whose corresponding output graphs $H$ have the property that their maximum degree is at most $c$ times their minimum degree, for some universal constant $c$.
Then Theorem \ref{thm:upper-bounds} holds in general.
\end{lemma}
The proof is deferred to Appendix \ref{app:reg}, as it is somewhat long but orthogonal to the main new ideas of this paper.
We shall assume in the rest of this section that $H$ (the graph output by Algorithm~\ref{alg:spanners}) is ``approximately regular'' in the sense of Lemma \ref{lem:3-reg}.
More specifically, throughout this paper, we let $D$ be a number such that all nodes in $H$ have degree $\Theta_k(D)$.

As usual, a \emph{walk} in $H$ is a sequence of nodes (possibly with repeats) in which every pair of adjacent nodes is connected by an edge.
A walk has length $i$ (also called an $i$-walk) if it contains $i+1$ nodes.
A walk is \emph{closed} if its first and last nodes are the same.
We shall also say that an edge $(u, v)$ belongs to a walk $w$ if the nodes $u, v$ appear in adjacent positions of $w$.
Let $c_{2k}$ denote the number of closed $2k$-walks in $H$.

\begin{definition}
If $H$ is the spanner produced by Algorithm \ref{alg:spanners} on an input graph $G = (V, E)$, then for an edge $(u, v) \in E$, the graph $H^{(u, v)} = (V, E')$ is defined as the subgraph of $H$ containing exactly the edges considered before $(u, v)$ during the greedy algorithm (not including $(u, v)$ itself).
\end{definition}
\begin{lemma} \label{lem:4cyc-completion}
$$c_4 = O(|E_H| \cdot fD)$$
\end{lemma}
\begin{proof}
We give the argument for vertex faults here; the argument for edge faults is essentially identical.
We shall argue that each edge added to $H$ during Algorithm \ref{alg:spanners} completes $O(fD)$ closed $4$-walks, which then implies the lemma by a simple union bound.

If we choose to add an edge $(u, v)$ to $H$, then $(u, v)$ is not $(3, f)$ protected in $H^{(u, v)}$.
%Thus there are not $f+1$ node-disjoint $u \leadsto v$ $3$-paths in $H^{(u, v)}$, and so (by the Min Cut/Max Flow Theorem \cite{?}) 
Thus, there is a set $F$ of $|F| \le f$ nodes such that every $u \leadsto v$ $3$-walk contains a node in $F$.
Additionally, for each node $x \in F$, there are $O(D)$ $u \leadsto v$ $3$-walks including the node $x$, since $3$ of the $4$ nodes on this walk are specified ($u$, $v$, and $x$) and the fourth node must be one of the $O(D)$ neighbors of $x$.
Applying a union bound, the total number of $u \leadsto v$ $3$-walks in $H^{(u, v)}$ is $O(fD)$.
See Figure \ref{fig:3-completion} for a picture associated with this argument.

For each such $u \leadsto v$ $3$-walk in $H^{(u, v)}$, we complete at most $8$ closed $4$-walks by adding $(u, v)$ to $H^{(u, v)}$ (specifically, each such $3$-walk corresponds to a single ``circular $4$-walk;'' we may then choose any of the $4$ nodes on this walk to serve as the start/end node of the appropriate closed walk, and it may be travelled in either of two directions).
The lemma follows.
\end{proof}

\begin{figure}[h]
\begin{center}

\begin{tikzpicture}
\draw [fill=black] (0, 0) circle [radius=0.15cm];
\draw [fill=black] (6, 0) circle [radius=0.15cm];
\node [below=0.2cm] at (0, 0) {$u$};
\node [below=0.2cm] at (6, 0) {$v$};

\draw [thick] (2, 0) ellipse (0.5cm and 2cm);
\draw [thick] (4, 0) ellipse (0.5cm and 2cm);

\node [below=0.2cm, align=center] at (1.8, -2) {First layer\\ size is\\ $|F| = f$};
\draw [fill=black] (2, 0) circle [radius=0.15cm];
\draw [fill=black] (2, 1) circle [radius=0.15cm];
\draw [fill=black] (2, -1) circle [radius=0.15cm];
\draw [thick] (0, 0) -- (2, 0);
\draw [thick] (0, 0) -- (2, 1);
\draw [thick] (0, 0) -- (2, -1);

\draw [thick] (2, 1) -- (4, 1.5);
\draw [thick] (2, 1) -- (4, 1);
\draw [thick] (2, 1) -- (4, 0.5);
\draw [thick] (2, 1) -- (4, 0);
\draw [thick] (2, 1) -- (4, -0.5);
\draw [thick] (2, 1) -- (4, -1);
\draw [thick] (2, 1) -- (4, -1.5);

\draw [thick] (2, 0) -- (4, 1.5);
\draw [thick] (2, 0) -- (4, 1);
\draw [thick] (2, 0) -- (4, 0.5);
\draw [thick] (2, 0) -- (4, 0);
\draw [thick] (2, 0) -- (4, -0.5);
\draw [thick] (2, 0) -- (4, -1);
\draw [thick] (2, 0) -- (4, -1.5);

\draw [thick] (2, -1) -- (4, 1.5);
\draw [thick] (2, -1) -- (4, 1);
\draw [thick] (2, -1) -- (4, 0.5);
\draw [thick] (2, -1) -- (4, 0);
\draw [thick] (2, -1) -- (4, -0.5);
\draw [thick] (2, -1) -- (4, -1);
\draw [thick] (2, -1) -- (4, -1.5);

\node [below=0.2cm, align=center] at (4.2, -2) {$O(fD)$ walks\\ from $u$ to\\ second layer};

\draw [thick] (6, 0) -- (4, 1.5);
\draw [thick] (6, 0) -- (4, 1);
\draw [thick] (6, 0) -- (4, 0.5);
\draw [thick] (6, 0) -- (4, 0);
\draw [thick] (6, 0) -- (4, -0.5);
\draw [thick] (6, 0) -- (4, -1);
\draw [thick] (6, 0) -- (4, -1.5);

\draw [fill=black] (4, 1.5) circle [radius=0.15cm];
\draw [fill=black] (4, 1) circle [radius=0.15cm];
\draw [fill=black] (4, 0.5) circle [radius=0.15cm];
\draw [fill=black] (4, 0) circle [radius=0.15cm];
\draw [fill=black] (4, -0.5) circle [radius=0.15cm];
\draw [fill=black] (4, -1) circle [radius=0.15cm];
\draw [fill=black] (4, -1.5) circle [radius=0.15cm];

\end{tikzpicture}
\caption{\label{fig:3-completion} In Lemma \ref{lem:4cyc-completion}, we show that each edge $(u, v)$ added to $H$ completes at most $O(fD)$ $u \leadsto v$ $3$-walks.  In this picture we have drawn the separating node set $F$ as coinciding with the first layer of the graph.}
\end{center}
\end{figure}

\begin{lemma} \label{lem:4cyc-cs}
$$|E_H| = O\left( n c_4^{1/4} \right)$$
\end{lemma}
\begin{proof}
The total number of $2$-walks in $H$ is $\Theta(n D^2)$, since we have $n$ choices for start node and $\Theta(D)$ neighbors for each node.
Each ordered pair of (possibly identical) $2$-walks $(w_1, w_2)$ in $H$ with the same start and end nodes corresponds uniquely to a closed $4$-walk, obtained by walking $w_1$ and then $w_2$ in reverse.
Thus, denoting by $k_{u, v}$ the number of $u \leadsto v$ $2$-walks, we may calculate
\begin{align*}
c_4 &= \sum \limits_{(u, v) \in V \times V} k_{u, v}^2\\
&\ge \frac{ \left(\sum \limits_{(u, v) \in V \times V} k_{u, v} \right)^2 }{n^2} & \text{Cauchy-Schwarz Inequality}\\
&= \frac{\Theta(n^2D^4)}{n^2}
= \Theta(D^4)
\end{align*}
and thus $nc_4^{1/4} \ge \Theta(nD) = |E_H|$.
\end{proof}

We are now ready to prove our upper bound for $k=2$.
\begin{proof} [Proof of Theorem \ref{thm:upper-bounds} for $k=2$]
Combining Lemmas \ref{lem:4cyc-completion} and \ref{lem:4cyc-cs}, we compute
\begin{align*}
|E_H| &= O\left( n \left(|E_H| \cdot fD \right)^{1/4} \right)\\
|E_H|^3 &= O\left(n^4 fD \right)\\
|E_H|^2 &= O\left(n^3 f \right) & \text{since } |E_H| = \Theta(nD)\\
|E_H| &= O\left(n^{3/2} f^{1/2} \right). \qedhere
\end{align*}
\end{proof}

\section{Lower Bound for $3$-Spanners} \label{sec:lower-2}

We prove Theorems \ref{thm:vertex-lbs} and \ref{thm:edge-lbs} in the special case $k=2$.
Incompressibility arguments for Theorems \ref{thm:vertex-incomp} and \ref{thm:edge-incomp} can be found in Appendix \ref{app:lbs}.
Our lower bounds for $k=2$ are unconditional, since the girth conjecture has been proved in this special case:
\begin{lemma} [e.g. Wenger \cite{wenger1991extremal}] \label{lem:4-girth}
For all $n$, there exist $n$-node graphs on $\Omega\left(n^{3/2}\right)$ edges without cycles of length $4$ or less.
\end{lemma}

We argue first for vertex faults:
\begin{proof} [Proof of Theorem \ref{thm:vertex-lbs}, case $k=2$]
Start with a graph $G$ from Lemma \ref{lem:4-girth}.
We construct a new graph $G'$ as follows.  Let $t = \lceil f/2 \rceil$.  We set $V_{G'} = V(G) \times [t]$, and let $E_{G'} = \{ \{(u,i), (v,j)\} : \{u,v\} \in E(G) \land i,j \in [t]\}$.  Let $G' = (V_{G'}, E_{G'})$.  Intuitively, we can think of $G'$ as being obtained by replacing each vertex of $G$ by a set of $t$ copies of the vertex, and each edge of $G$ is replaced by a complete bipartite graph between the two sets of copies.  We will prove Theorem~\ref{thm:vertex-lbs} by proving that the only $f$ VFT $3$-spanner of $G'$ is itself, and that it has the required number of edges.

We first claim that the only $f$ VFT $3$-spanner of $G'$ is $G'$ itself.  To see this, suppose that $H$ is a subgraph of $G'$ which does not contain some edge $\{(u,i), (v,j)\} \in E_{G'}$.  Let $F = \{(u,\ell) : \ell \neq i\} \cup \{(v, \ell) : \ell \neq j\}$, i.e., we let the fault set be all copies of $u$ except for $(u,i)$ and all copies of $v$ except $(v,j)$.  Note that $|F| \leq f$.  Now consider the shortest path from $(u,i)$ to $(v,j)$ in $H \setminus F$.  Let this path be $(u,i) = (x^0, i^0), (x^1, i^1), (x^2, i^2), \dots, (x^p, i^p) = (v,j)$.  Note that for any $0 \leq a \leq p-1$, it cannot be the case that $x^a = u$ and $x^{a+1} = v$, since no such edges exist in $H \setminus F$.  Thus $u = x^0, x^1, \dots, x^p = v$ is a walk from $u$ to $v$ in $G$ which does not use the edge $\{u,v\}$.  By adding $\{u,v\}$ to this path, we get a cycle of length at most $p+1$.  Since $G$ has girth at least $5$, this implies that $p \geq 4$.  Thus in $H \setminus F$ the distance between $(u,i)$ and $(v,j)$ is at least $4$, while in $G' \setminus F$ they are at distance $1$.  So $H$ is not an $f$ VFT $3$-spanner of $G'$, and hence the only $f$ VFT $3$-spanner of $G'$ is $G'$ itself.

%
%Since this is the shortest path and the only remaining copy of $u$ is $(u,i)$, we know that $x^a \neq u$ for all $a > 0$.  Similarly, since it is a shortest path and each edge of $G$ has been replaced by a complete bipartite graph, we know that if $0 < a < b$ then $x^a \neq x^b$ (or else there would be a shorter path which would go directly from $(x^{a-1}, i^{a-1})$ to $(x^b, i^b)$).  Hence $u = x^0, x^1, \dots, x^p = v$ must be a simple $u-v$ path in $G$, and so when combined with the edge $\{u,v\}$ forms a simple cycle in $G$ of length $p+1$.  But $G$ has girth at least $2k+2$, and hence $p \geq 2k+1$.  Thus $H$ is not a $(2k-1)$-spanner of $G' \setminus F$, and thus is not a $f$ VFT $(2k-1)$-spanner of $G'$.  

So $G'$ is the only $f$ VFT $3$-spanner of itself, and it remains only to analyze its size.  Clearly $|V_{G'}| = t |V(G)|$, and by Lemma~\ref{lem:4-girth} we know that $|E(G)| \geq \Omega(|V(G)|^{3/2})$.  Thus
\begin{align*}
|E_{G'}| &= \Omega(\left(|E(G)| t^2\right) = \Omega\left(t^2 |V(G)|^{3/2}\right) = \Omega\left(t^2 \cdot \left(\frac{|V_{G'}|}{t}\right)^{3/2}\right) \\
&= \Omega\left(f^{1/2} |V_{G'}|^{3/2}\right)
\end{align*}
as claimed.
\end{proof}

And next for edge faults:
\begin{proof} [Proof of Theorem \ref{thm:edge-lbs}, case $k=2$]

We construct a new graph $G'$ exactly as in the VFT case.  As before, let $H$ be a subgraph of $G'$ missing some edge $\{(u,i), (v,j)\}$.  Let $F = \{\{(u,i), (v, \ell)\} : \ell \in [t] \setminus \{j\}\} \cup \{\{(u,\ell), (v,j)\} : \ell \in [t] \setminus \{i\}\}$ be the fault set, and note that $|F| \leq f$.  In other words, we fail every edge from $(u,i)$ to copies of $v$ except for the edge to $(v,j)$, and similarly we fail all edges from $(v,j)$ to copies of $u$ except for $(u,i)$.  So in $G' \setminus F$ the edge $\{(u,i), (v,j)\}$ is still present, but it is not in $H \setminus F$.  Moreover, in $H \setminus F$ there are no paths of length at most $3$ from $(u,i)$ to $(v,j)$ that use as an intermediate node any other copy of $u$ or copy of $v$.  Hence any path of length at most $3$ must be of the form $(u,i), (x, a), (y,b), (v,j)$ where $u, x, y, v$ are all distinct (possibly with either $x$ or $y$ missing).  This implies that $u,x,y,v$ form either a $3$- or a $4$-cycle in $G$, which contradicts Lemma~\ref{lem:4-girth}.  Thus $H \setminus F$ is not a $3$-spanner of $G' \setminus F$, so $H$ is not a $f$ EFT $3$-spanner of $G'$.

Thus $G'$ is the only $f$ EFT $3$-spanner of itself.  Using the same analysis as in the VFT case, we get that $|E(G')| \geq \Omega(f^{1/2} |V(G')|^{3/2})$, proving the theorem.
\end{proof}

\section{Overview: Upper Bounds for Larger $k$ \label{sec:ub-overview}}

For simplicity, we will focus on the case $k=3$ in this overview.
The most natural attempt to extend our upper bounds to $k=3$ goes as follows: first one generalizes Lemma \ref{lem:4cyc-completion}, and then one generalizes Lemma \ref{lem:4cyc-cs} by considering $3$-walks rather than $2$-walks, and then one combines the two statements as before.
Unfortunately, the upper bounds implied by this approach are quite weak.
One easily generalizes Lemma \ref{lem:4cyc-cs} to show that
$$|E_H| = O \left( nc_{6}^{1/6} \right),$$
and one easily generalizes Lemma \ref{lem:4cyc-completion} to show
$$c_{6} = O\left( |E_H| \cdot fD^3 \right).$$ 
However, plugging these equations into the proof used before, we actually find \emph{no improvement} in the previous upper bound: we still get
$$|E_H| = O \left( n^{3/2} f^{1/2} \right)$$
while we should hope for something much better.

Let us now informally sketch a method for improving the generalization of Lemma \ref{lem:4cyc-completion}.
Recall that, from Lemma \ref{lem:4cyc-completion} itself, we have $c_4 = O\left( |E_H| \cdot fD \right)$ (note that this proof holds even in the setting $k=3$).
Thus, \emph{on average}, each node $v$ participates in $O(fD^2)$ closed $4$-walks.
Additionally, we have $\Theta(D^2)$ $2$-walks starting at $v$.
Hence each of these two walks (on average) is responsible for creating at most $O(f)$ closed $4$-walks in $H$, and so (on average) there are only $O(f)$ $2$-walks between any two nodes.

Assume for a moment that we could move from an average to worst-case version of this statement, and assume a hard limit of $O(f)$ $2$-walks between any two nodes.
This fact could be used to improve our generalization of Lemma \ref{lem:4cyc-completion}, reasoning as follows.
Suppose we add the edge $(u, v)$ to $H^{(u, v)}$ during Algorithm \ref{alg:spanners}.
Then, there is a set $F$ of $f$ nodes that lie on any $u \leadsto v$ $5$-walk.
Consider any such node $x$, and assume without loss of generality that $x$ is in the first half of any $u \leadsto v$ $5$-walk.
We then count the number of $u \leadsto v$ $5$-walks including $x$ as follows: there are $D^2$ ways to choose the first four nodes on the walk (since $u, x$ are two of them), and there are $f$ ways to choose the last two nodes on the walk (since $v$ is the last node on the walk and we have specified the fourth node on the walk, and we have assumed that there are only $O(f)$ possible $2$-walks between these nodes).
Hence there are $O(fD^2)$ $u \leadsto v$ $5$-walks including $x$; applying a union bound over $F$, we have $O(f^2 D^2)$ $u \leadsto v$ $5$-walks in total, and this leads to an improved bound
$$c_6 = O \left( |E_H| \cdot (fD)^2 \right).$$

\begin{figure}[h]
\begin{center}

\begin{tikzpicture}
\draw [fill=black] (0, 0) circle [radius=0.15cm];
\draw [fill=black] (10, 0) circle [radius=0.15cm];
\node [below=0.2cm] at (0, 0) {$u$};
\node [below=0.2cm] at (10, 0) {$v$};
\node at (5, 0) {$\cdots$};

\draw [thick] (2, 0) ellipse (0.5cm and 2cm);
\draw [thick] (4, 0) ellipse (0.5cm and 2cm);
\draw [thick] (6, 0) ellipse (0.5cm and 2cm);
\draw [thick] (8, 0) ellipse (0.5cm and 2cm);

\node [below=0.2cm, align=center] at (1.8, -2) {First layer\\ size is\\ $|F| = f$};
\draw [fill=black] (2, 0) circle [radius=0.15cm];
\draw [fill=black] (2, 1) circle [radius=0.15cm];
\draw [fill=black] (2, -1) circle [radius=0.15cm];
\draw [thick] (0, 0) -- (2, 0);
\draw [thick] (0, 0) -- (2, 1);
\draw [thick] (0, 0) -- (2, -1);

\draw [thick] (2, 1) -- (4, 1.5);
\draw [thick] (2, 1) -- (4, 1);
\draw [thick] (2, 1) -- (4, 0.5);
\draw [thick] (2, 1) -- (4, 0);
\draw [thick] (2, 1) -- (4, -0.5);
\draw [thick] (2, 1) -- (4, -1);
\draw [thick] (2, 1) -- (4, -1.5);

\draw [thick] (2, 0) -- (4, 1.5);
\draw [thick] (2, 0) -- (4, 1);
\draw [thick] (2, 0) -- (4, 0.5);
\draw [thick] (2, 0) -- (4, 0);
\draw [thick] (2, 0) -- (4, -0.5);
\draw [thick] (2, 0) -- (4, -1);
\draw [thick] (2, 0) -- (4, -1.5);

\draw [thick] (2, -1) -- (4, 1.5);
\draw [thick] (2, -1) -- (4, 1);
\draw [thick] (2, -1) -- (4, 0.5);
\draw [thick] (2, -1) -- (4, 0);
\draw [thick] (2, -1) -- (4, -0.5);
\draw [thick] (2, -1) -- (4, -1);
\draw [thick] (2, -1) -- (4, -1.5);

\node [below=0.2cm, align=center] at (6, -2) {$O(fD^2)$ walks\\ from $u$ to\\ third layer};

\draw [fill=black] (6, 1.5) circle [radius=0.15cm];
\draw [fill=black] (6, 1) circle [radius=0.15cm];
\draw [fill=black] (6, 0.5) circle [radius=0.15cm];
\draw [fill=black] (6, 0) circle [radius=0.15cm];
\draw [fill=black] (6, -0.5) circle [radius=0.15cm];
\draw [fill=black] (6, -1) circle [radius=0.15cm];
\draw [fill=black] (6, -1.5) circle [radius=0.15cm];

\draw [fill=black] (8, 0) circle [radius=0.15cm];
\draw [fill=black] (8, 1) circle [radius=0.15cm];
\draw [fill=black] (8, -1) circle [radius=0.15cm];

\draw [thick] (6, 0) -- (8, 1);
\draw [thick] (6, 0) -- (8, 0);
\draw [thick] (6, 0) -- (8, -1);
\draw [thick] (10, 0) -- (8, 1);
\draw [thick] (10, 0) -- (8, 0);
\draw [thick] (10, 0) -- (8, -1);

\node [below=0.2cm, align=center] at (10, -2) {$f$ ways\\ to complete a\\ walk from the\\ third layer to $v$};

\end{tikzpicture}

\caption{Assuming a maximum of $O(f)$ $2$-paths between any two nodes, each edge $(u, v)$ added to $H$ completes at most $O(f^2 D^2)$ $u \leadsto v$ $3$-walks.  In this picture we have again drawn the separating node set $F$ as coinciding with the first layer of the graph.}
\end{center}
\end{figure}
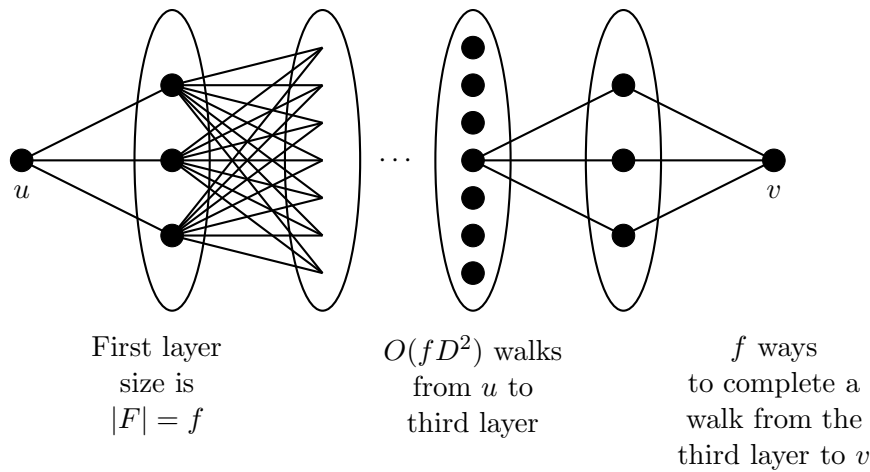

This is precisely the bound on $c_6$ needed to prove Theorem \ref{thm:upper-bounds} for $k=3$!
Of course, it ``only'' remains to justify our assumption that there are no more than $O(f)$ $2$-walks between any two nodes, while we have only proved this statement in the average case.
This requires new machinery.
At a high level, we accomplish this by carefully restricting our attention to a subset of the paths in $H$; in particular, we ``block'' $2$-paths that exceed the desired $O(f)$ bound, and then use something morally similar to the above proof to show that each edge $(u, v)$ completes only $O(f^2 D^2)$ ``unblocked'' closed $6$-walks.
It is a careful balancing act to define these terms in the right way while still having an analog of Lemma \ref{lem:4cyc-cs} that applies when only ``unblocked'' closed $6$-walks are considered.
However, with enough precision, it can be done. 

Some additional technical work not mentioned here is required to avoid a dependence of the form e.g. $O(\text{poly}(k))$ in our upper bound.
We defer a discussion of this point to the appendix.

\section{Overview: Lower Bounds for Larger $k$}

\input{lower-bound}

\bibliography{references}
	\bibliographystyle{plain}

\appendix

\section{Regularizing $H$ \label{app:reg}}

Before diving into our main analysis, it will be convenient as before to assume that $H$ is ``approximately regular.''
We justify this assumption by the following argument.

%First, we may assume without loss of generality that $G=H$; if Algorithm \ref{alg:spanners} discards any edge in $G$ without adding it to $H$, then we may instead simply assume that this edge was never in $G$ in the first place, and this will not change the final size of $H$.
%
%Second, we observe that Theorem \ref{thm:upper-bounds} is equivalent to the following statement: for any graph $G$ on $n' \le n$ nodes, the spanner $H$ returned by Algorithm \ref{alg:spanners} on $G$ has average degree $O_k \left( n^{1/k} f^{1 - 1/k} \right)$ (it is immediate that this statement and Theorem \ref{thm:upper-bounds} imply each other).
%
%We now process $G$ as follows to gain regularity.
%Fix $D$ as its initial average degree.
%Delete all nodes from $G$ whose degree is $D/4$ or less.
%Then, for all remaining nodes, 
%
%---------------------

\begin{lemma} \label{lem:k-reg}
Suppose that Theorem \ref{thm:upper-bounds} holds for all graphs $G$ whose corresponding output graphs $H$ have the property that their maximum degree is at most $c$ times their minimum degree, for some universal constant $c$.
Then Theorem \ref{thm:upper-bounds} holds in general.
\end{lemma}

\begin{proof}
We shall prove this lemma by showing the contrapositive.
We suppose that $G = (V, E)$ is an $n$-node graph that serves as a counterexample to Theorem \ref{thm:upper-bounds} (i.e. its corresponding output graph $H$ has $n^{1 + 1/k} f^{1 - 1/k} \omega(1)^k$ edges).
Our goal is then to show that there is a subgraph $G' \subseteq G$ on at least $n^{1/(2k)}$ nodes that also serves as a counterexample to Theorem \ref{thm:upper-bounds}, and the maximum degree of $G'$ is larger than its minimum degree by a factor of $O_k(1)$.
Hence, in the proof of Theorem \ref{thm:upper-bounds} we may assume this approximate regularity condition.
We will then rule out the existence of a family of $G'$, and this rules out the general existence of $G$.
Note that the constraint that $G'$ has at least $n^{1/(2k)}$ nodes is used to ensure that an infinite family of $G$ implies an infinite family of $G'$; any super-constant function would work equally well for our application of this lemma.

First, we may assume that $G = H$; if our algorithm discards any edge in $G$, we may simply delete this edge from $G$ and it is clear that this deletion will not change the corresponding output graph $H$.
Let $n$ be the number of nodes in $G=H$ and let $D$ be its average degree. 
Partition the nodes of $G$ into two sets: the set $A$ of nodes that have degree $cD$ or less, and the set $B$ of nodes that have degree more than $cD$ (for some constant $c = c_k$ that will be chosen later).
By Markov's inequality we have $|B| \le n/c$.
Partition the edges of $G$ into three sets: $E_A$ where both endpoints are in $A$, $E_B$ where both endpoints are in $B$, and $E_{AB}$ that has one endpoint in each.
Compare the sizes of $E_A, E_B, E_{AB}$; keep all edges in the largest set and discard all edges in the other two (ties may be broken arbitrarily).
We now split into cases based on which set survived.

\paragraph{(a) $E_A$ Survived.}
We now have a maximum degree of $cD$ in the remaining graph, while its average degree is still $\Omega(D)$ (since we have only discarded a constant fraction of the edges).
Let $D'$ be the new average degree of $E_A$, and repeatedly delete all nodes from $G$ that have degree $D'/4$ or less until all nodes have degree at least $D' / 4$.
Note that only $n \cdot D'/4$ edges (which is at most half the remaining edges in $G$) will be removed in this way.
The remaining subgraph $G'$ is now sparser than the original graph $G$ by only a constant factor and its maximum and minimum degrees are both $\Theta(D)$.
Additionally, the number of nodes $n'$ in $G'$ is $n' = \Omega(|E|)/\Theta(D) = \Omega(n)$.
Thus $G'$ satisfies the lemma.
%\mike{Don't we still need to argue that there are at least $n^{1/(2k)}$ nodes? --Mike}
%\gre{Yes, thanks, added. --Greg}

\paragraph{(b) $E_B$ Survived.}
We now obtain a subgraph $G'$ by deleting all nodes in $A$ (which now have degree $0$).
The average degree in the remaining subgraph $G'$ is at least $cD/3$, and it has $1/c$ as many nodes as before.
It is not necessarily the case that $G'$ is approximately regular, but we may now recurse the argument on $G'$ to enforce this property (it remains to be proved that this recursion terminates before removing too many nodes from $G$).
Note that $G'$ has higher average degree and fewer nodes than $G$, so it still serves as a counterexample graph to our claimed bound.

\paragraph{(c) $E_{AB}$ Survived.}
The average degree of the nodes in $A$ is now at least $D/6$, since we have deleted at most $1/3$ of the edges since setting $D$ and (since the graph is now bipartite) each edge is incident on one node in $A$.
Our next step is to delete all nodes from $A$ except for the $|B|$ nodes in $|A|$ with the highest remaining degree.
The average degree in the remaining subgraph $G'$ is then still at least $D/6$.
We also have $2|B| \le \frac{2}{c} \cdot n$ nodes remaining in $G'$.
We then once again recurse the analysis on the remaining subgraph (and we will prove shortly that this recursion terminates before removing too many nodes).\\
%Note that we have $D = \Omega(n^{1/k})$, and thus if we choose $c = \omega(6^k)$, then $G'$ still serves as a counterexample to Theorem \ref{thm:upper-bounds}.\\

In the latter two cases, we recurse the analysis on a subgraph $G'$ of the original graph.
We will now show that, if $c$ is chosen to be suitably large, then the recursion must eventually terminate in case (a) while the graph still has $\text{poly}(n)$ nodes (where $n$ is the number of nodes in the original graph $G$).
Naturally, throughout the recursion, the average degree $D^*$ in the subgraph being considered cannot exceed the number of nodes $n^*$ in that subgraph.
However, after $k$ rounds of recursion that avoid case (a), we have
$$D^* \ge \frac{D}{6^k} \ge  \frac{n^{1/k}}{6^k}$$
and
$$n^* \le n \cdot (2/c)^k$$
We thus have
\begin{align*}
\frac{n^{1/k}}{6^k} &\le n \cdot (2/c)^k\\
(c/12)^k &\le n^{1 - 1/k}\\
k \log(c/12) &\le (1 - 1/k) \log n \le \log n
\end{align*}
Thus, by choosing $c$ sufficiently large, we have that $k \le \frac{\log n}{c'}$ (for some constant $c'=\log(c/12)$ that can be made arbitrarily large by choice of $c$).
This means that the recursion bottoms out at depth $\frac{\log n}{c'}$, at which point the average degree of the graph is
$$D^* \ge \frac{n^{1/k}}{3^{\log n / c'}} \ge \frac{n^{1/k}}{n^{1/c''}} = n^{1/k - 1/c''}$$
where $c''$ is another constant that can be made arbitrarily large by pushing $c'$ arbitrarily large.
Choosing $c'' \ge 2k$ we have $D^* \ge n^{1/(2k)}$ and so the recursion must terminate while the graph still has at least $n^{1/(2k)}$ nodes, as claimed.

 Now if we set $c''=2k$, then we are setting $c=12\cdot 9^k$. If the recursion bottoms out after $i\leq k$ levels in case (a), then after applying   case (a), the average degree is $D^{*}\geq D/(2\cdot 6^i)$ and the number of nodes is $n^*\leq n(2/c)^i$. We will explicitly show that the subgraph $G^*$ we have obtained is a counterexample, provided $G$ was. We know that in $G$, the number of edges is $Dn/2\geq n^{1+1/k}f^{1-1/k}Q$ where $Q\geq\omega(1)$. We will show that in $G^*$, the number of edges is $\geq (n^*)^{1+1/k}f^{1-1/k}Q$, and hence $G^*$ is also a counterexample.

First, notice that $n^*\leq n(2/c)^i = n(1/(6\cdot 9^k)^i \leq n/(2^k\cdot 6^{ki})$. Thus, $(n^*)^{1/k}\leq n^{1/k}/(2\cdot 6^i)$. Now, since $D\geq 2n^{1/k}f^{1-1/k}Q$ (since $G$ is a counter example), we get that $D^*\geq D/(2\cdot 6^i)\geq 2f^{1-1/k} Q n^{1/k}/(2\cdot 6^i)\geq 2f^{1-1/k} (n^*)^{1/k} Q$, and hence the number of edges in $G^*$ is $D^* n^*/2\geq f^{1-1/k} (n^*)^{1+1/k} Q$, and so $G^*$ is also a counterexample of nontrivial size.

\end{proof}

Hence, it suffices to refute the existence of ``approximately regular'' graphs $G'$ that violate Theorem \ref{thm:upper-bounds}.
By contrapositive of Lemma \ref{lem:k-reg}, this would imply that no \emph{general} graph may violate Theorem \ref{thm:upper-bounds}.
Note that this argument implies Lemma \ref{lem:3-reg} by plugging in $k=2$.

As we proceed, we will assume that the maximum and minimum degrees of the input graph $G$ differ by a factor of $O_k(1)$.
In fact, as before we may assume that $G = H$, so we may assume this same approximate regularity property for $H$.
Moreover, we will assume that all nodes in $H$ have degree in the interval $[D, \psi_k D]$ for some parameter $D$ that we carry into the following proofs. From the argument in Lemma \ref{lem:k-reg}, we have that $\psi_k\leq 2^{O(k)}$.

\section{Upper Bounds for $k \ge 3$ \label{app:k-ub}}

We now prove Theorem \ref{thm:upper-bounds} for larger $k$.
For convenience, we will assume that all edges in $H$ have unique weights, and are thus considered by our algorithm in a consistent order (or it suffices that ties between equally-weighted edges are broken in some consistent fashion).

\subsection{Definitions, Notation, and some Intuition}
First let us introduce some new notation related to walks in $H$.
We write $\overline{w}$ to denote the reverse of a walk $w$, we write $w^i$ to denote the $i^{th}$ node in the walk $i$ (indexing from $0$), and for two walks $w_1, w_2$ such that the last node of $w_1$ equals the first node of $w_2$, we write $w_1 w_2$ to denote their concatenation in the natural way.

As mentioned in Section \ref{sec:ub-overview}, it is easy to prove a suitable generalization of Lemma \ref{lem:4cyc-cs} but hard to prove a suitable generalization of Lemma \ref{lem:4cyc-completion}.
The difficulty in proving Lemma \ref{lem:4cyc-completion} for $k > 2$ is in handling pair of nodes $u, v$ for which the number of $u \leadsto v$ $k$-walks is much larger than the average number of $k$-walks over all nodes pairs in the graph.
Our solution is to ``block'' these irregular parts of the graph, essentially throwing away a constant fraction of the available walks in the graph in exchange for a guarantee of ``regularity'' on the ones that survive.
Specifically, we will shortly define a set $\bee$ of walks in $H$, where each $b \in \bee$ is called a \emph{blockade}.
A walk $w$ in $H$ is \emph{blocked} by some $\bee$ if there is a blockade $b \in \bee$ such that $b$ or $\overline{b}$ is a sub-walk of $w$ (we use the notation $b \subseteq w$ or $\overline{b} \subseteq w$).

%\mike{The use of $\subseteq$ seems a little odd to me -- clearly since the statement is for $b$ or $\overline{b}$ the ordering matters, but $\subseteq$ isn't usually defined for orderings.  Fine to keep it, since we use it later, but maybe we should at least state what it means in the context of (ordered) walks. --Mike}

After appropriately designing $\bee$, it is tempting to generalize Lemma \ref{lem:4cyc-completion} by counting unblocked closed walks.
A technical detail of introducing blocks is that this is no longer quite the right object to count.
Let us define:
\begin{definition}
An \emph{$i$-walk meet} is an ordered pair of (possibly identical) $i$-walks with the same start and end points.
\end{definition}
The proof of Lemma \ref{lem:4cyc-completion} works by observing that the number of $2$-walk meets is the same as the number of closed $4$-walks, up to constant factors.
However, it is not necessarily the case that the number of \emph{unblocked} $i$-walk meets is within a constant factor of the number of \emph{unblocked} closed $2i$-walks.
Indeed, the number of $i$-walk meets can be much larger.
This happens because there are $2i$ different $i$-walk meets corresponding to each closed $2i$-walk, and one can imagine that some of these are blocked but others are unblocked.
With this in mind, our proof works by counting $i$-walk meets directly rather than passing through any attempt to count closed walks, as before.
\begin{figure}
\begin{center}
\begin{tikzpicture}
\draw [thick] (0, 0) circle [radius=3cm];
\draw [fill=black] (0, -3) circle [radius=0.15cm];
\node [below=0.15cm] at (0, -3) {$x$};
\draw [fill=black] (-3, 0) circle [radius=0.15cm];
\node [left=0.15cm] at (-3, 0) {$a$};
\draw [fill=black] (0, 3) circle [radius=0.15cm];
\node [above=0.15cm] at (0, 3) {$b$};
\draw [fill=black] (3, 0) circle [radius=0.15cm];
\node [right=0.15cm] at (3, 0) {$c$};

\draw [red,ultra thick,domain=0:180] plot ({3*cos(\x)}, {3*sin(\x)});
\draw [blue,ultra thick,domain=265:95, ->] plot ({2.7*cos(\x)}, {2.7*sin(\x)});
\draw [blue,ultra thick,domain=-85:85, ->] plot ({2.7*cos(\x)}, {2.7*sin(\x)});

\end{tikzpicture}
\end{center}
\caption{In this picture, the closed $4$-walk $w$ starting and ending at $x$ is blocked by the path $(a, b, c) \in \bee$, but both paths in the $2$-walk meet $[(x, a, b), (x, c, b)]$ associated with it are still unblocked.  This explains why the number of unblocked $2$-walk meets can potentially be much larger than the number of unblocked closed $4$-walks.}
\end{figure}
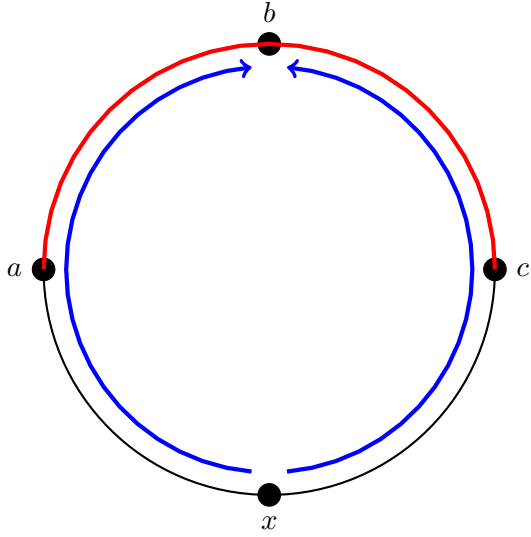

We shall write $W^{\bee}, X^{\bee}$ to denote the set of walks in $H$ that are unblocked and blocked by $\bee$, respectively.
We will frequently omit the superscript when it is simply $\bee$ (we never refer to $W$ in the absence of a blockade set).
We also use the following modifiers on the sets $\bee, W, X$:
we denote by $\bee_i, \bee_{\le i}$ the subset of walks in $\bee$ of length $i$ or length at most $i$, respectively, and we denote by $\bee[s \leadsto t]$ to denote the subset of walks in $\bee$ with endpoints $s, t$ (and similar notation is used on $W, X$).
Note that $X_i^{\bee}$ is not to be confused with $\bee_i$.%; in fact, $X_i^{\bee}$ can be nonempty even if the maximum length walk in $\bee$ is $i-1$.

Similarly, we denote by $M_i^{\bee}$ is the set of all unblocked $i$-walk meets in $H$, and $M_i^{\bee}[s \leadsto t]$ is the subset of these where both walks in the meet start at $s$ and end at $t$.
We will also frequently suppress the superscript, although there is always an implicit $\bee$.

Our goal is to establish the following inequality:
$$\Theta(D)^{2k} \leq \left|M_k^{\bee_{\leq k-1}}\right| \leq \Theta_{k}\left(\left|E_H\right|\cdot (fD)^{k-1}\right)$$
which will then imply an upper bound on $|E_H|$ by some straightforward algebra.
The left-hand inequality is shown in Lemmas \ref{lem:unblocked-count} and \ref{lem:kcyc-cs}. %\mike{Lemma~\ref{lem:unblocked-count} doesn't actually seem to reference $M_k^{\bee_{\leq k-1}}$.  Do you mean Lemma~\ref{lem:kcyc-cs}? --Mike}.
The high level idea here is that, if we only block a small enough constant fraction (dependent on $k$) fraction of the $i$-walks for all $i\leq k$, then the number of unblocked $k$-walk meets is still nearly as large as one would expect without blockades (using the Cauchy-Schwartz Inequalty).
The main technical effort goes into the right-hand inequality.
This is proved using a delicate interlacing inductive argument, in which two complementary counting arguments are proved for the case $i=2$ and then used to boost each other up to $i=k$.
Specifically, we show an interesting interplay between (1) bounding the number of 
unblocked $i$-walk meets $M_i^{\bee\leq i-1}$ and (2) choosing $i$-length blockades to bound the number of unblocked $i$-walks $W_i^{\bee\leq i}[u \leadsto v]$ between any given pair of nodes $u,v$ (Lemma  \ref{lem:building-blocks}).
The latter bound is then used to bounding the number of unblocked $i+1$-walk meets $M_{i+1}^{\bee\leq i}$, and so on.

\subsection{Lower Bound for $|M_k|$ (Generalization of Lemma \ref{lem:4cyc-cs})}

The following lemma generalizes the fact used in the $k=2$ setting that we have $\Theta(nD^2)$ $2$-walks in $H$ to consider.
\begin{lemma} \label{lem:unblocked-count}
There is some $\phi_k > 0$ (independent of $n, f$) such that the following property holds: if $\bee_i$ contains at most a $\phi_k$ fraction of all $i$-walks in $H$ for all $i$, then
$$|W_i| \ge n \Theta(D)^i \quad \text{for all } i \le k.$$
\end{lemma}
\begin{proof}
First note that we have
$$ nD^i \le |W_i \cup X_i| \le n(\psi_kD)^i$$
\emph{total} walks in $H$ (recall that by our regularization procedure, we can assume that all vertices have degree in $[D, \psi_k D]$).
We now upper bound $X_i$.
Each $x \in X_i$ has a subpath $x \supseteq b \in \bee$.
We may thus upper bound $X_i$ by the following (somewhat loose) estimate:
$$|X_i| \le \sum \limits_{j=2}^{i} |\bee_j| \cdot (2\psi_kD)^{i-j},$$
since each blocked path $x \in X_i$ may be obtained via $i - |b|$ extensions of some blockade $b \in \bee_{\le i}$ by adding an edge to its front or back, and there are up to $2\psi_k D$ ways to extend any given walk by $1$ edge on either end.
Since $|\bee_j| \le \phi_k n(\psi_k D)^j$  we have
$$|X_i| \le \sum \limits_{j=2}^{i} n \phi_k \cdot (2)^{i-j} (\psi_k D)^i \le n\phi_k (2\psi_k D)^i.$$
Hence, choosing $\phi_k \le \frac{1}{(4\psi_k)^k}$ (this is overkill), we have
$$|X_i| \le n (D/2)^i$$
And so
$$|W_i| \ge |W_i \cup X_i| - |X_i| \ge nD^i - n(D/2)^i = n\Theta(D)^i$$
and the lemma follows.
\end{proof}

With this, we can prove:
\begin{lemma} [Generalization of Lemma \ref{lem:4cyc-cs}] \label{lem:kcyc-cs}
If $|W_k| = n \Theta(D)^k$, then $|E_H| = O_k\left(n\left|M_k\right|^{1/(2k)}\right)$.
\end{lemma}
\begin{proof}
The proof is essentially identical to that of Lemma \ref{lem:4cyc-cs}.
We have 
\begin{align*}
\left|M_k \right| &= \sum \limits_{(u, v) \in V \times V} \left|W_k [u \leadsto v]\right|^2\\
&\ge \frac{ \left(\sum \limits_{(u, v) \in V \times V} \left|W_k [u \leadsto v]\right| \right)^2 }{n^2} & \text{Cauchy-Schwarz Inequality}\\
&= \frac{ \left(n\Theta(D)^k \right)^2 }{n^2}\\
&= \Theta(D)^{2k}
\end{align*}
and so
\begin{align*}
\left|M_k\right|^{1/(2k)} &\ge D\\
n\left|M_k\right|^{1/(2k)} &\ge nD\\
n\left|M_k\right|^{1/(2k)} &= \Omega_k\left(|E_H|\right) \\
\end{align*}
which implies the lemma.
\end{proof}

\subsection{Upper Bound for $|M_k|$ (Generalization of Lemma \ref{lem:4cyc-completion})}
We now begin to work towards a generalization of Lemma \ref{lem:4cyc-completion}.
The argument will be inductive in nature.
Lemma \ref{lem:4cyc-completion} will essentially serve as the base case, and the inductive step is split across the next two lemmas.

\begin{lemma} \label{lem:meet-count}
Suppose that $\left|W_j[u \leadsto v]\right| = O_k\left(f^{j-1}\right)$ for all $u, v$ and all $j < i$ for some fixed $i \le k$.
Then
$$\left|M_i \right| =O_k\left(|E_H| \cdot (fD)^{i-1}\right).$$
\end{lemma}
\begin{proof}
For a given $i$-walk meet $(w_1, w_2)$, we will say that the \emph{shift} $s_x$ of a node $x \in w_1 \cup w_2$ is its first position in $w_1 \overline{w_2}$ (so $0 \le s_x < 2i$).
%\mike{What are $w_1$ and $w_2$?  Arbitrary $i$-walks?  Or they have to share endpoints? --Mike}.
Similarly, the shift $s_{(u, v)}$ of an edge $(u, v)$ is the first position of this edge in $w_1 \overline{w_2}$ (where the first edge in $w_1 \overline{w_2}$ is indexed from $0$).

As before, when we choose to add any $(u, v)$ to $H$, it is not $(2k-1, f)$ protected in $H^{(u, v)}$ and so there is a set $F$ of $|F| \le f$ nodes (with $u, v \notin F$) such that every $u \leadsto v$ walk of length $2i-1$ intersects some node $x \in F$.
Note that each $(w_1, w_2) \in M_i$ that is completed by the addition of the edge $(u, v)$ to $H$ corresponds to a $u \leadsto v$ walk of length $2i-1$ for $i \leq k$ (obtained by joining the common start/endpoint of $w_1 \overline{w_2}$ to create a \emph{circular} walk, and then removing the edge $(u, v)$ from this walk).
Thus $x \in w_1 \cup w_2$ for some $x \in F$.
Our proof strategy is: we fix a shift $s_{(u, v)}$ for the edge $(u, v)$, we fix a node $x \in F$ and a shift value $s_x$ for $x$, and we will show that the addition of the edge $(u, v)$ to $H^{(u, v)}$ completes only $O_k\left(f^{i-2} D^{i-1}\right)$ $i$-walk meets including the node $x$ in which $(u, v), x$ have their prescribed shift values.
Since there are $O_k(f)$ possible choices of $x \in F, 0 \le s_x, s_{(u, v)} \le 2k$, this implies
$$\left|M_i\right| = O_k\left(|E_H| \cdot (fD)^{i-1}\right)$$
by a simple union bound over $(u, v) \in E_H$ and over the possible choices of $x, s_x, s_{(u, v)}$.

In the rest of this proof we will assume $0 \le s_{(u, v)} \le i-1$, and so $(u, v) \in w_1$ (we also assume without loss of generality that $u$ precedes $v$ in $w_1$).
The proof for the remaining case $i \le s_{(u, v)} < 2i$ is identical with the roles of $w_1, w_2$ swapped.
We will denote by $M'_i$ the subset of $M_i$ satisfying the criteria mentioned thus far (i.e. with $(u, v)$ at shift $s_{(u, v)}$ and $x$ at shift $s_x$).
We split into two cases, depending on the value of $s_x$.

\paragraph{(a) Suppose $0 \le s_x \le i+1$.}
Let $P$ be the set of $i+1$-walks in $H^{(u, v)}$ that include node $u$ in position $s_{(u, v)}$, node $v$ in position $s_{(u, v)} + 1$, and node $x$ in position $s_x$ ($P$ can even include blocked walks).
We next count $|P|$.
Supposing $s_x > s_{(u, v)}$, after fixing the node $u$ in position $s_{(u, v)}$, there are $(\psi_k D)^{s_{(u, v)}}$ ways to choose a prefix for the walk, and there are $(\psi_k D)^{i - 1 - s_{(u, v)}}$ ways to choose a suffix for the walk, since $v, x$ have fixed positions following $u$.
Hence $|P| = (\psi_k D)^{s_{(u, v)}} \cdot (\psi_k D)^{i - 1 - s_{(u, v)}} = O_k(D^{i-1})$. 
If instead we have $s_x < s_{(u, v)}$ then the argument and bound are identical (with a factor of $\psi_k D$ shifted from the prefix to the suffix).

%Each $(w_1, w_2) \in M'_i$ is uniquely determined by the walk $w_1 \overline{w_2}$, and for each $w_1 \overline{w_2}$ there is exactly one path $p \in P$ that forms a prefix.
For each $p = \left(p^0, \ldots, p^{i+1}\right) \in P$, let $Q^p$ be the set of suffixes such that for each $q \in Q^p$, we have $qp = w_1\overline{w_2}$ for some $(w_1, w_2) \in M'_i$.
We now count $\left|Q^p\right|$.
Note that each $q \in Q^p$ has the form $q = w_1 \overline{w_2} \setminus p$, and so $\overline{q}$ is an $i-1$-subwalk of $w_2$; since $w_2$ is unblocked by definition of $i$-walk meets, it follows that $q$ is unblocked as well.
Thus $Q^p \subseteq W_{i-1}[p^{i+1} \leadsto p^0]$, and so
$$\left| Q^p \right| \le \left| W_{i-1}[p^{i+1} \leadsto p^0] \right| = O_k\left(f^{i-2}\right)$$
We then have
$$\left|M'_i\right| = \sum \limits_{p \in P} \left| Q^p \right| \le |P| \cdot O_k\left(f^{i-2}\right) = O_k\left(f^{i-2} D^{i-1}\right)$$
as claimed.

\paragraph{(b) Instead suppose $i+2 \le s_x < 2i$.}
This implies that $x \notin w_1$ so $x \in w_2$.
The proof is now only a slight tweak on the previous case.
We define $P$ to be the set of $i$-walks (not $i+1$ walks as before) in $H^{(u, v)}$ that have $u$ in position $s_{(u, v)}$ and $v$ in position $s_{(u, v)} + 1$.
By an identical counting argument as in the previous case, we have $|P| = O_k(D^{i-1})$.
%As before, for each $(w_1, w_2) \in M_i$, there is exactly one path $p \in P$ that is a prefix of $w_1\overline{w_2}$ (in fact $p = w_1$).
As before, for each fixed $p = \left(p^0, \dots, p^{i}\right) \in P$ we define $Q^p$ as the set of suffixes such that for each $q \in Q^p$, we have $qp = w_1 \overline{w_2}$ for some $(w_1, w_2) \in M_i$, and our goal is to count $\left|Q^p\right|$.
Note that we then have $\overline{q} = w_2$ and so $q$ is unblocked.
Unlike before, we have $x \in q$ so we may split each $q \in Q^p$ into two subpaths over $q$.
Specifically, let
$$Q_1 := \left\{q_1 \subseteq q \ \mid \ q_1^0 = p^i, q_1^{s_x - i} = x, q \in Q^p\right\}$$
and
$$Q_2 := \left\{q_2 \subseteq q \ \mid \ q_2^0 = x, q_2^{2i - s_x} = p^0, q \in Q^p\right\}$$
(in other words, for any $q \in Q^p$, we split $q$ over the node $x$ and add its prefix to $Q_1$ and its suffix to $Q_2$).
Since each $q \in Q^p$ is unblocked, we also have that each $q_1 \in Q_1, q_2 \in Q_2$ is unblocked.
Thus $Q_1 \subseteq W_{s_x - i}[p^i \leadsto x]$ and $Q_2 \subseteq W_{2i - s_x}[x \leadsto p^0]$.
We then have
$$\left| Q^p \right| \le \left|Q_1 \right| \cdot \left|Q_2 \right| \le O_k\left(f^{s_x - i - 1}\right) \cdot O_k\left(f^{2i - s_x - 1}\right) = O_k\left(f^{i-2}\right)$$
and so, as before,
$$\left|M'_i\right| = \sum \limits_{p \in P} \left|Q^p \right| = |P| \cdot O_k\left(f^{i-2}\right) = O_k \left( f^{i-2} D^{i-1} \right)$$
as claimed.
\end{proof}

We next show:
\begin{lemma} \label{lem:building-blocks}
Let $2 \le i \le k$.
Suppose that
$$\left|M_i^{\bee_{\le i-1}}\right| = O_k\left(|E_H| \cdot (fD)^{i-1}\right).$$
Then there exists a set $\bee_i$ of $i$-length blockades such that
$$\left| W_i^{\bee_{\le i-1} \cup \bee_i} [u \leadsto v] \right| = O_k\left(f^{i-1}\right) \qquad \text{for all } u, v \in V$$ 
and $\bee_i$ contains at most a $\phi_k$ fraction of all $i$-walks in $H$, for any desired constant $\phi_k > 0$ depending only on $k$.
\end{lemma}
\begin{proof}
We construct $\bee_i$ iteratively as follows: initially $\bee_i = \emptyset$; we then repeatedly choose the node pair $u, v$ that maximizes $\left| W_i^{\bee_{\le i-1} \cup \bee_i} [u \leadsto v] \right|$ (ties may be broken arbitrarily), choose any walk $b \in  W_i^{\bee_{\le i-1} \cup \bee_i} [u \leadsto v]$, and add $b$ to $\bee_i$ (note that the sets $W_i^{\bee_{\le i-1} \cup \bee_i} [u \leadsto v]$ shrink throughout this process as we grow $\bee_i$).
Repeat until no more walks may be added to $\bee_i$ without destroying the property that $\bee_i$ contains at most a $\phi_k$ fraction of all $i$-walks in $H$.

We now argue that $\bee_i$ satisfies the lemma.
Observe that $\left|M_i[s \leadsto t]\right| = \left|W_i[s \leadsto t]\right|^2$ for any $s, t \in V$.
Suppose towards a contradiction that $\bee_i$ does \emph{not} satisfy the lemma; that is,
$$\left| W_i^{\bee_{\le i-1} \cup \bee_i} [u \leadsto v] \right| = \omega_k\left(f^{i-1}\right) \qquad \text{for some } u, v \in V.$$
Then while we iteratively build $\bee_i$, each time we choose a node pair $s, t$ that maximizes\\$\left| W_i^{\bee_{\le i-1} \cup \bee_i} [s \leadsto t] \right|$, we must have
$$\left| W_i^{\bee_{\le i-1} \cup \bee_i} [s \leadsto t] \right| = \omega_k\left(f^{i-1}\right).$$

Thus, when we add some $b \in W_i^{\bee_{\le i-1} \cup \bee_i} [s \leadsto t]$ to $\bee_i$, the size of $\left(\left|W_i[s \leadsto t]\right|-1\right)^2$ falls from $\left|W_i[s \leadsto t]\right|^2$ to
$$\left(\left|W_i[s \leadsto t]\right|-1\right)^2 = \left|W_i[s \leadsto t] \right|^2 - \Omega\left(\left|W_i[s \leadsto t]\right|\right) = \left|W_i[s \leadsto t] \right|^2 - \omega_k\left(f^{i-1}\right)$$
and so the total number of $i$-walk meets in $H$ falls by an additive $\omega_k\left(f^{i-1}\right)$ term in each iteration of building $\bee_i$.
By Lemma \ref{lem:unblocked-count} there are $n\Theta_k(D)^i$ total $i$-walks in $H$, so
$$\left| \bee_i \right| = \phi_k \cdot n\Theta_k(D)^i = n\Theta_k(D)^i.$$
Hence the number of $i$-walk meets falls by
$$\omega_k \left(f^{i-1} \right) \cdot n\Theta_k(D)^i = \omega_k\left( |E_H| \cdot (fD)^{i-1}\right)$$
from start to finish of the process of building $\bee_i$ (since $|E_H| = n\Theta_k(D)$).
This is a contradiction, since we assumed that initially $\left|M_i^{\bee_{\le i-1}}\right| = O_k\left( |E_H| \cdot (fD)^{i-1}\right)$.
Thus $\bee_i$ satisfies the lemma statement.
\end{proof}

\begin{lemma} [Generalization of Lemma \ref{lem:4cyc-completion}] \label{lem:kcyc-completion}
There is a blockade set $\bee$ for which
$$\left|M_k\right| = O_k\left( \left| E_H \right| (fD)^{k-1} \right)$$
and for all $i \le k$, the set $\bee_i$ contains at most a $\phi_k$ fraction of all $i$-walks in $H$, for any desired constant $\phi_k > 0$ depending only on $k$.
\end{lemma}
\begin{proof}
We will inductively show that two statements hold for all $2 \le i \le k$:
\begin{itemize}
\item (The lemma statement) There is a blockade set $\bee_{\le i-1}$ (with maximum blockade length $i-1$) such that (1) $\bee_j$ contains at most a $\phi_k$ fraction of all $j$-walks in $H$ for all $j \leq i-1$, and (2) $\left|M_i^{\bee_{\leq}i-1}\right| = O_k\left(\left|E_H\right| (fD)^{i-1} \right)$.
%\mike{I don't really understand the relationship between $i$ and $j$ here -- it looks like $j$ could be larger than $i$?  But that seems weird. --Mike}

\item (An auxiliary statement) We have $\left| W_j^{\bee} [u \leadsto v] \right| = O_k\left(f^{j-1}\right)$ for all $u, v \in V$ and $j \le i-1$.
\end{itemize}

We first argue the base case $i=2$.
For the first condition (lemma statement), we appeal to Lemma \ref{lem:4cyc-completion} to argue that the bound holds even when $\bee_{\le 1} = \emptyset$.
Specifically, note that without blockades, each closed $4$-walk in $H$ corresponds in the natural way to $\Theta(1)$ $i$-walk meets.
Thus
$$\left|M_2^{\emptyset}\right| = O\left( \left| E_H \right| \cdot fD \right).$$
We also note that the auxiliary statement holds trivially when $i=2$; it simply states that there is at most $1$ edge between any two nodes.

Now we argue the inductive step, assuming that both of the above properties hold for $i$.
The first part of the induction (the lemma statement) coincides with the premise of Lemma \ref{lem:building-blocks}, so we may find a set $\bee_i$ as in the conclusion of Lemma \ref{lem:building-blocks}.
That is, we have
$$\left| W_i^{\bee_{\leq i-1} \cup \bee_i} [u \leadsto v] \right| = O_k\left(f^{i-1}\right) \qquad \text{for all } u, v \in V.$$
This statement (together with the auxiliary inductive hypothesis) proves the auxiliary statement for the case $i+1$.
The auxiliary statement for the case $i+1$ also coincides with the premise of Lemma \ref{lem:meet-count}; applying this lemma, we have
$$\left|M_{i+1}^{\bee_{\leq i-1} \cup \bee_i}\right| = O_k\left( \left| E_H \right| \cdot \left(fD\right)^i \right).$$
This is the lemma statement for the case $i+1$, which completes the inductive hypothesis.
Hence the first part of our inductive hypothesis (the lemma statement) holds in the case $i=k$, which completes the proof. 
\end{proof}

\subsection{Proof of Theorem \ref{thm:upper-bounds}}

It is now a matter of algebra to complete the proof of Theorem \ref{thm:upper-bounds}.
\begin{proof} [Proof of Theorem \ref{thm:upper-bounds}]
Let $\bee$ be a blockade set as in Lemma \ref{lem:kcyc-completion}.
By Lemma \ref{lem:kcyc-completion}, we have
$$\left|M_k\right| = O_k\left( \left| E_H \right| (fD)^{k-1} \right).$$
Additionally, since $\bee_i$ has at most a $\phi_k$ fraction of all $i$-walks for each $i$, we have $\left|W_i\right| \ge n\Theta(D)^k$ (by Lemma \ref{lem:unblocked-count}) and so
$$\left|E_H\right| = O_k\left( n\left|M_k\right|^{1/(2k)} \right)$$
by Lemma \ref{lem:kcyc-cs}.
Combining these, we compute
\begin{align*}
\left| E_H \right| &= O_k \left( n \left( \left|E_H \right| \left(fD\right)^{k-1} \right)^{1/(2k)} \right)\\
\left| E_H \right|^{2k-1} &= O_k \left( n^{2k} \left(fD\right)^{k-1} \right)\\
\left| E_H \right|^k &= O_k \left( n^{k+1} f^{k-1} \right) & \text{since } \left|E_H\right| = \Theta_k(nD)\\
\left| E_H \right| &= O_k \left( n^{1 + 1/k} f^{1 - 1/k} \right)
\end{align*}
as claimed.
\end{proof}

\input{app-lower}

\end{document}

%% file: abstract.tex
% !TEX root = newftspanners.tex
\begin{abstract}
A \emph{$k$-spanner} of a graph $G$ is a sparse subgraph $H$ whose shortest path distances match those of $G$ up to a multiplicative error $k$.
In this paper we study spanners that are resistant to faults.
A subgraph $H \subseteq G$ is an $f$ \emph{vertex fault tolerant (VFT)} $k$-spanner if $H \setminus F$ is a $k$-spanner of $G \setminus F$ for any small set $F$ of $f$ vertices that might ``fail.''
One of the main questions in the area is: what is the minimum size of an $f$ fault tolerant $k$-spanner that holds for all $n$ node graphs (as a function of $f$, $k$ and $n$)?
This question was first studied in the context of geometric graphs [Levcopoulos et al. STOC '98, Czumaj and Zhao SoCG '03] and has more recently been considered in general undirected graphs [Chechik et al. STOC '09, Dinitz and Krauthgamer PODC '11].

In this paper, we settle the question of the optimal size of a VFT spanner, in the setting where the stretch factor $k$ is fixed.
Specifically, we prove that every (undirected, possibly weighted) $n$-node graph $G$ has a $(2k-1)$-\emph{spanner} resilient to $f$ vertex faults with $O_k(f^{1 - 1/k} n^{1 + 1/k} )$ edges, and this is fully optimal (unless the famous Erd{\" o}s Girth Conjecture is false).
Our lower bound even generalizes to imply that no \emph{data structure} capable of approximating $\dist_{G \setminus F}(s, t)$ similarly can beat the space usage of our spanner in the worst case.
We note that all previous upper bounds carried an almost \emph{quadratic} (or worse) dependence on $f$, whereas the dependence in our tight bound is \emph{sublinear}.
To the best of our knowledge, this is the first instance in fault tolerant network design in which introducing fault tolerance to the structure increases the size of the (non-FT) structure by a sublinear factor in $f$.
Another advantage of this result is that our spanners are constructed by a very natural and simple greedy algorithm, which is the obvious extension of the standard greedy algorithm used to build spanners in the non-faulty setting.

We also consider the \emph{edge fault tolerant (EFT)} model, defined analogously with edge failures rather than vertex failures.
We show that the same spanner upper bound applies in this setting.
Our data structure lower bound extends to the case $k=2$ (and hence we close the EFT problem for $3$-approximations), but it falls to $\Omega(f^{1/2 - 1/(2k)} \cdot n^{1 + 1/k})$ for $k \ge 3$.
We leave it as an open problem to close this gap.
\end{abstract}

%% file: intro.tex
% !TEX root = newftspanners.tex

A {\em spanner} (\cite{PelegS:89,PelegU:89}) of a graph is a subgraph that approximately preserves its shortest path metric.
%Graph spanners (\cite{PelegS:89,PelegU:89}) are sparse 
%subgraphs that approximate the distances between any pair
%of vertices with respect to the original graphs. 
More formally, a subgraph $H = (V, E' \subseteq E)$ is a $t$-spanner of a graph $G= (V, E)$ if
$$\dist_H(u,v) \leq t \cdot \dist_G(u,v) \quad \text{ for all } u,v \in V$$
($t$ is called the \emph{stretch} of the spanner).
% (where $\dist_G(u,v)$ is the shortest-path distance between $u$ and $v$ in $G$).
%The value $t$ is known as the \emph{stretch} of the spanner.
Spanners were introduced by Peleg and Ullman~\cite{PelegU:89} and Peleg and Sch{\"{a}}ffer~\cite{PelegS:89}, and have a wide range of applications in routing \cite{PelegU:89-routing}, synchronizers \cite{awerbuch1990network}, broadcasting \cite{awerbuch1991cient,peleg2000distributed}, distance oracles \cite{thorup2005approximate}, graph sparsifiers \cite{kapralov2012spectral}, and even preconditioning of linear systems \cite{elkin2008lower}. 
The most common objective in spanners research is to achieve the best possible existential size-stretch trade-off.
% as fast as possible. \gre{we don't do this though :)}
Most notably, a landmark result of Alth\"ofer et al.~\cite{AlthoferDDJS:93} proved that for any integer $k \geq 1$, every graph $G=(V, E)$ has a $(2k-1)$-spanner $H \subseteq G$ with $O(n^{1+1/k})$ edges, and moreover, there exist graphs for which this size-stretch tradeoff cannot be improved (if we assume the girth conjecture of Erd\H{o}s \cite{erdHos1964extremal}). 
In fact, their existentially optimal upper bound was obtained via an extremely simple and natural greedy construction algorithm:
consider the edges of $G$ in non-decreasing order of their weight and add an edge $\{u,v\}$ to the current spanner $H$ if and only if $\dist_H(u,v)>(2k-1) w(u,v)$.
It is easy to verify that this algorithm never creates cycles of length $2k$ or less in $H$, and simple folklore upper bounds imply that any graph of girth $> 2k$ has $O(n^{1+1/k})$ edges. 

A crucial aspect of real-life systems that is not captured by the standard notion of spanners is the possibility of failure.
If some edges (e.g., communication links) or vertices (e.g., computer processors) fail, what remains of the spanner might not still approximate the distances of what remains of the original graph.
This motivates the notion of \emph{fault-tolerance} for spanners.
The canonical model was first introduced by Levcopoulos, Narasimhan, and
Smid \cite{levcopoulos1998efficient} (who happened to work in the geometric setting):
a subgraph $H$ is an $f$ vertex (edge) fault tolerant $t$-spanner for $G$ if
$$\dist_{H \setminus F}(u,v)\leq t \cdot \dist_{G \setminus F}(u,v) \quad \text{ for every }u,v \in V \text{ and } F \subseteq V (F \subseteq E), |F|\leq f.$$
In other words, a fault tolerant spanner $H$ contains a spanner for $G\setminus F$ for every set $F$ of $f$ nodes/edges that could fail.

The question of whether it is possible to construct a sparse fault tolerant
spanner for an arbitrary undirected weighted graph (rather than a geometric graph) was raised by Czumaj and Zhao~\cite{czumaj2004fault}.
This was answered in the affirmative by Chechik, Langberg, Peleg and Roditty \cite{ChechikLPR:10} who gave the first results on fault-tolerant spanners for general graphs. 
They presented constructions of an $f$-vertex fault tolerant $(2k-1)$-spanner of an $n$-node graph $G$ of size $O(f^2 k^{f+1} \cdot n^{1+1/k}\log^{1-1/k}n)$, and an
$f$-edge fault tolerant $(2k-1)$-spanner of $G$ of size $O(f\cdot n^{1+1/k})$.  Hence, Chechik et al.~\cite{ChechikLPR:10} showed that introducing tolerance to $f$ edge faults costs us an extra factor of $f$ in the size of the spanner, while introducing tolerance to $f$ vertex faults costs us a factor of $f^2 k^{f+1}$ in the size (compared to the size of a non-fault tolerant spanner of the same stretch).  
Dinitz and Krauthgamer~\cite{DinitzK:11} later improved the edge bound for vertex faults to $O\left(f^{2-\frac{1}{k}}n^{1+\frac{1}{k}}\log{n}\right)$.
%This improved the dependence on the number of faults from exponential to sub-quadratic.

While this prior work provided a significant lead on the problem, it left behind two interesting knowledge gaps.
\begin{enumerate}
\item It is open whether these dependencies are ``right;'' no nontrivial lower bounds for the problem are yet published, and so improvements to these upper bounds are conceivable.
For example, it is natural to wonder: {\em is it possible to pay only an extra factor of $f$ and still achieve vertex fault-tolerance (as is possible for edge failures)?  Even more fundamentally: can we pay even \emph{less} than $f$ and still achieve $f$-fault tolerance?}

\item One of the major pros of the textbook (non-faulty) greedy spanner construction of Alth\"ofer et al.~\cite{AlthoferDDJS:93} is the simplicity and obvious correctness of the algorithm.
The prior work on fault tolerant spanners exhibits new and interesting techniques, but cannot reasonably be viewed as an analog or extension of the classic greedy spanner.
It is thus open whether a similar degree of algorithmic simplicity can be achieved in the fault tolerant model.
\end{enumerate}

In this paper, we present new upper and lower bounds for fault tolerant spanners that directly address both of these issues.

\paragraph{Our Contribution.}
We answer the above questions affirmatively.
First, we show that we can construct vertex and edge-fault tolerant spanners that cost only $o(f)$ more than their non-fault tolerant counterparts.  
%Slightly more formally, we consider the natural extension of the Alth{\"{o}}fer et al.~\cite{AlthoferDDJS:93} greedy algorithm to the fault tolerant setting, and show:
\begin{theorem} [Main Result, Upper Bound] \label{thm:upper-bounds}
Let $G = (V, E, w)$ be an undirected graph with real edge weights and no negative-weight cycles.
Let $k\geq 1$ be a fixed integer.
For any (possibly non-constant) positive integer $f$, $G$ has an $f$-VFT $(2k-1)$-spanner on $O_k(f^{1 - 1/k} n^{1 + 1/k})$ edges. 
The same bounds can be achieved for $f$-EFT spanners.
\end{theorem}
%\mike{also true for EFT, right?}

%The big-O in the sparsity bound hides a $2^{O(k)}$ factor, but for fixed $k$ this is a constant.
%We leave it as an open problem to determine whether this exponential dependence on $k$ is necessary.

Interestingly, the construction algorithm behind Theorem \ref{thm:upper-bounds} (Algorithm \ref{alg:spanners}) is indeed the natural generalization of Alth\"ofer et al.~\cite{AlthoferDDJS:93}: we simply consider the edges of the graph in non-decreasing order of weight, and we add an edge $(u, v)$ to the current spanner $H$ if and only if $\dist_{H \setminus F}(u, v) > (2k-1)w(u, v)$ for any possible fault set $F$.
Correctness of the algorithm is once again trivial.
However, this time it is highly nontrivial to prove an upper bound the density of the final spanner $H$, and it is the main endeavor of this paper to establish this.

The $O_k$ in the sparsity bound in Theorem \ref{thm:upper-bounds} hides a $2^{O(k)}$ factor (which is a constant for fixed $k$).
We leave it as an open problem to determine whether this exponential dependence on $k$ is necessary.

To the best of our knowledge, Theorem~\ref{thm:upper-bounds} is the first construction of any fault tolerant graph structure whose size has a \emph{sublinear} dependence in $f$. Perhaps the most basic fault tolerant structures are those preserving connectivity, where the connected components of $H \setminus F$ are the same as for $G \setminus F$ for any fault set $F$ of size at most $f$.  Even for this much simpler requirement, the existing constructions of FT connected subgraphs \cite{parter2016fault} pay a factor of at least $f$ in the size compared to the size of the non-fault tolerant structure. 
%\mer{There is no paper that explicitly studies the FT-connectivity problem, but the known algorithm is very similar to the edge FT-spanner algorithm of Chechik et al. Maybe we can also add the $s-t$ $f$-connectivity problem. Here, it is easy to have a case where one needs to pay $\Omega(f \cdot \dist(s,t,G))$ for having $f$ node disjoint paths, where $\dist(s,t,g)$ is the cost of the non-FT solution.}
In fact, a common belief in fault tolerant network design is that paying a factor $f$ in the size of the fault tolerant structure is the best one can achieve \cite{parter2016fault}.  Showing that this belief is wrong is a core contribution of this paper. We hope that it will motivate finding other graph predicates where requiring fault tolerance only costs us a sublinear factor of $f$ compared to the size of the non-faulty structure. 

For the case of vertex faults, we complement Theorem~\ref{thm:upper-bounds} with a matching lower bound: 
\begin{theorem} [Lower Bound, Vertex Faults] \label{thm:vertex-lbs}
For any positive integers $k$ and $f$, assuming the Erd\"{o}s Girth Conjecture \cite{erdHos1964extremal}, there exist infinite families of undirected unweighted $n$-node graphs on $\Omega(f^{1-1/k}\cdot n^{1 + 1/k})$ edges for which any $f$ VFT $(2k-1)$-spanner $H$ must contain all the edges of $G$. 
\end{theorem}
The Girth Conjecture is widely believed and widely used as a basis for lower bounds in spanners research.
It has been confirmed for $k \in \{1, 2, 3, 5\}$ \cite{wenger1991extremal} (and thus our lower bounds are unconditional for these values of $k$), and it is open for all other values of $k$.

For edge fault-tolerance we can prove the same lower bound for the special case of $k=2$, but for larger stretch values we give a weaker lower bound.
We leave it as an open question to close this gap for $k \ge 3$.
\begin{theorem} [Lower Bound, Edge Faults] \label{thm:edge-lbs}
For any positive integers $k$ and $f$, assuming the Erd\"{o}s Girth Conjecture \cite{erdHos1964extremal}, there exist infinite families of undirected unweighted $n$-node graphs on
\[
\begin{cases}
\Omega \left( f^{1/2} n^{3/2} \right) & \text{ when } k=2 \\
\Omega \left( f^{1/2 - 1/(2k)} n^{1 + 1/k} \right) & \text{ when } k \ge 3\\
\end{cases}
\]
edges for which any $f$ EFT $(2k-1)$-spanner $H$ must contain all the edges of $G$.
\end{theorem}

In fact, by applying some standard tricks to our new lower bound constructions, we can generalize to prove \emph{strong incompressibility} theorems:
\begin{theorem} [Strong Incompressibility, Vertex Faults] \label{thm:vertex-incomp}
For any positive integers $k \ge 2$ and $f$, assuming the Erd{\" o}s Girth Conjecture \cite{erdHos1964extremal}, there is no algorithm that can process $n$-node graphs $G = (V, E)$ into a data structure $\dee_G$ on $o\left(f^{1 - 1/k} n^{1 + 1/k}\right)$ bits such that $\dee_G$ can answer queries (in any amount of time) of the form $(s, t, F)$, where $s, t \in V, F \subseteq V, |F| \le f$, with a value $\widehat{\dist_{G \setminus F}}(s, t)$ satisfying
$$\dist_{G \setminus F}(s, t) \le \widehat{\dist_{G \setminus F}}(s, t) \le (2k-1) \cdot \dist_{G \setminus F}(s, t).$$
\end{theorem}

\begin{theorem} [Strong Incompressibility, Edge Faults] \label{thm:edge-incomp}
For any positive integers $k \ge 2$ and $f$, assuming the Erd{\" o}s Girth Conjecture \cite{erdHos1964extremal}, there is no algorithm that can process $n$-node graphs $G = (V, E)$ into a data structure $\dee_G$ on
\[
\begin{cases}
o \left( f^{1/2} n^{3/2} \right) & \text{ when } k=2 \\
o \left( f^{1/2 - 1/(2k)} n^{1 + 1/k} \right) & \text{ when } k \ge 3\\
\end{cases}
\]
bits such that $\dee_G$ can answer queries (in any amount of time) of the form $(s, t, F)$, where $s, t \in V, F \subseteq E, |F| \le f$, with a value $\widehat{\dist_{G \setminus F}}(s, t)$ satisfying
$$\dist_{G \setminus F}(s, t) \le \widehat{\dist_{G \setminus F}}(s, t) \le (2k-1) \cdot \dist_{G \setminus F}(s, t).$$
\end{theorem}

Note that the spanner $H \subseteq G$ promised by Theorem \ref{thm:upper-bounds} functions as a data structure $\dee_G$ for these purposes, and thus these bounds are essentially optimal.\footnote{More precisely, these data structure lower bounds are optimal up to a $\log n$ factor, since the spanner on $|E|$ edges can take up to $|E| \cdot \log n$ bits to encode.}

\paragraph{Discussion of our techniques.} 
%\gre{I moved some content from this section into the above, because I wanted to sell the greedy algorithm a little more directly.}
As mentioned previously, the fault-tolerant spanners in Theorem \ref{thm:upper-bounds} are constructed by the natural generalization of the greedy spanner of Alth{\"{o}}fer et al.~\cite{AlthoferDDJS:93} into the fault-tolerant setting.
%To construct our fault-tolerant spanners in Theorem~\ref{thm:upper-bounds} we use a very simple and natural generalization of the Alth{\"{o}}fer et al.~\cite{AlthoferDDJS:93} greedy algorithm: Consider the edges in non-decreasing order of their weight, and add an edge $(u,v)$ to the current FT spanner $H$ if there is some set of $f$ faults $F$ such that there is no path on at most $2k-1$ edges in $H\setminus F$ connecting $u$ and $v$.
%\mike{This isn't actually what the algorithm is -- not being protected (the criterion in the algorithm) is necessary for a bad fault set to exist, but not sufficient.  That is, our algorithm might actually add edges that aren't strictly necessary according to the definition of fault-tolerance.  This is Lovasz's old ``Mengerian Theorems for Paths of Bounded Length" paper --Mike.}
%\gre{Wow, very nice catch.  I had assumed these were the same notion.  I changed the definition of "protected" to the version in which $F$ is a node separator, so I think this should now be correct ... (the main algorithm in the appendix was assuming that we had node separators). --Greg} \vir{Cool! I went through the paper and fixed up all mentions of the protectedness to refer to the correct definition. --v}
%In other words, if $(u,v)$ is not added to $H$, then no matter what $f$ faults $F$ we consider, the distance in $H\setminus F$ between $u$ and $v$ is at most $(2k-1)w(u,v)$.
It is easy to see that the spanner $H$ built by the greedy algorithm is an $f$-fault tolerant $(2k-1)$-spanner, both for vertex and for edge faults.
Our main technical contribution in this paper is in the analysis of the number of edges $m$ in $H$.
We overview this here.

In the non-faulty greedy spanner algorithm, the output spanner $H$ is shown to have \emph{no} cycles of length at most $2k$.
In our algorithm, we focus on counting the closely related notion of closed walks of length exactly $2k$.
Our output spanner $H$ might have many closed walks of length $2k$; let this number be $C$.
At a very high level, we attempt to give both a lower bound $L(m,n,f)$ and an upper bound $U(m,n,f)$ on $C$, and then derive an upper bound on $m$ in terms of $n, f$ by rearranging the inequality $L(m,n,f)\leq U(m,n,f)$.

First, it is not too hard to obtain a lower bound $L(m,n,f)$ on the number of closed $2k$-walks of length up to $2k$ using an argument based on the Cauchy-Schwartz inequality: we have $C = \Omega((m/n)^{2k})$ as long as $m\geq 100kn^{1+1/k}$, say (this lower bound actually holds for $2k$-cycles as well, as shown by Morris and Saxton~\cite{morrissaxton}; a simple argument for $k=2$ can be found in~\cite{BringmannGSW16}).
Based on this lower bound, it turns out that the desired bound on $m$ follows if we can upper bound the number of short cycles in $H$ by $O(m(fm/n)^{k-1})$, as we then would get 
$(m/n)^{2k}\leq O(f^{k-1} m^k/n^{k-1})$ which implies $m\leq O(f^{1-1/k} n^{1+1/k})$, as desired.

Unfortunately, the only thing that our greedy VFT spanner algorithm gives us is that when an edge $\{u, v\}$ is added to $H$, there is a set $F$ of $f$ nodes such that every short enough $u \leadsto v$ path includes a node in $F$.
A major difficulty is to go from this disjoint paths condition to upper-bounding the number of short closed walks.
To overcome this difficulty, we take several steps.
First we present a regularization technique that allows us to assume that the output spanner $H$ is roughly regular.
This rough regularity allows us to say that the number of $i$-walks starting from any given vertex is roughly $D^i$, where $D=\Theta(m/n)$ is the average degree of $H$.
%This makes our arguments much simpler.
%However, the main issue still remains, that we need to have a dependence on $f$ in our upper bound on the number of cycles.
For the special case $k=2$, this makes our upper bound argument quite clean (see Section 3): for any pair of nodes $u,v$, when the edge $\{u, v\}$ is added, the number of $3$-walks from $u$ to $v$ is $\Theta(fD)$, since all these paths must intersect one of the $f$ nodes $x \in F$, and (by rough regularity) the remaining node in the walk must be one of the $\Theta(D)$ possible neighbors of $x$.
Thus, by a union bound over the edge set of $H$, the total number of closed $4$-walks is $O(fDm)$ and the desired upper bound on $m$ follows.

For larger $k$, however, a difficulty arises in extending this argument.
Let us generously suppose that all the $2k-1$ walks connecting $u$ and $v$ go through one of $f$ neighbors of $v$.
Can one still argue that the total number of $2k-1$-walks connecting $u$ and $v$ (when $\{u, v\}$ is added to $H$) is at most $O(m(fm/n)^{k-1})$?
A naive extension of the previous argument would only give an upper bound of $O(m f D^{2k-2})$ which is quite far from what we want.

To obtain a better argument, we introduce quite a bit of machinery. For instance, 
instead of bounding the number of closed $2k$-walks, we show that it is sufficient to bound the number of pairs of $k$-walks that meet at the same endpoints.
We also introduce the notion of \emph{blockades}, which allows us to only count certain types of walks, and allows us to push through a delicate inductive argument that finally achieves the correct upper bound of $O(m(fm/n)^{k-1})$ that leads to our main upper bound theorem.

\paragraph{Additional Related Work.}
Constructing fault tolerant spanners for geometric graphs (or graphs from ``simple" metric spaces such as doubling metrics) has been further studied in \cite{lukovszki1999new,levcopoulos2002improved,abam2009region,solomon2014hierarchical,chan2015sparse}.
The study of \emph{robust} geometric spanners, where removing few vertices from the graph harms only a small number of other vertices, was initiated in \cite{bose2013robust}.

In general graphs, the construction of sparse fault tolerant subgraphs has received a significant amount of recent attention \cite{peleg2009good,parter2016fault}. 
Construction of purely additive fault tolerant spanners were studied in \cite{BraunschvigCPS:15,BiloGGSP:15,BGPV17}. Concerning \emph{exact} distances from a single (or few) sources, \cite{ParterP:13} introduced the notion of FT-BFS structures that contain a BFS tree tree from $s$ in $G\setminus \{e\}$ for every failing edge $e \in E(G)$. 
They showed an upper bound of $O(n^{3/2})$ edges and provided a matching lower bound graph example. FT-BFS structures avoiding $2$ faults with optimal size were given in \cite{Parter:15}. 
Approximate versions of FT-BFS structures (where the structure is allowed to have a stretch on the $s \times V$ distances) were studied in \cite{khanna2010approximate,ParterP:14,Bilo14-esa,BiloG0P16}.

A natural data structure analog of fault-tolerant subgraphs are \emph{distance sensitivity oracles}, which are small data structure that are used to answer queries of the form: ``what is the distance between $s$ and $t$ in $G$ when a set of $F$ edges fail"?
There is a long line of literature on constructing distance sensitivity oracles \cite{demetrescu2008oracles,BernsteinK:08,bernstein2009nearly,grandoni2012improved, baswana2012single,WeimannY:13,duan2016improved,Bilo16-esa} and related structures such a fault tolerant routing schemes \cite{gavoille2008compact,chechik2011fault} and labeling schemes \cite{abraham2016forbidden}. 
 Recently, \cite{ChechikCFK17} provided an efficient construction of distance sensitivity oracles that support $f=O(\log n/\log \log n)$ many faults with polylogarithmic query time. In an another breakthrough, \cite{duan2017connectivity} showed a connectivity sensitivity oracle that supports $f \in [1,n]$ vertex failures with $O(f m \log n)$ space, update time $O(g^2)$ and query time $O(g)$ where $g\leq f$ is the number of actual faults.   

Turning to reachability in directed graphs, \cite{BaswanaCR:16} showed that for any number of faults $f\geq 1$, there is a subgraph $H \subseteq G$ with $O(2^f n)$ edges that preserves the reachability from $s$ after the failure of any $f$ edges. They also showed that this upper bound is existentially tight.

%% file: technical.tex
\section{Preliminaries}
We begin with a formal definition of fault tolerant spanners:
\begin{definition} [Fault Tolerant Spanners]
Let $G = (V, E, w)$ be an undirected weighted graph without negative weight cycles.
We say that a subgraph $H \subseteq G$ over the same vertex set is an $f$ vertex (edge) fault tolerant $t$-spanner of $G$ if, for any set $F$ of $f$ vertices (edges), we have
$$\dist_{H \setminus F}(u, v) \le t \cdot \dist_{G \setminus F}(u, v) \quad \text{ for all } u, v \in V$$
where $G \setminus F$ and $H \setminus F$ denote these graphs with the set $F$ of vertices (edges) removed.
\end{definition}

%This conjecture is confirmed for $k \in \{1, 2, 3, 5\}$ \cite{wenger1991extremal} and open for all other $k\geq 7$.\\
We abbreviate \emph{vertex fault tolerant} and \emph{edge fault tolerant} by VFT and EFT, respectively.
We will construct our spanners in Theorem \ref{thm:upper-bounds} using the following natural construction algorithm:
\begin{definition}
In a graph $H$, a pair of nodes $(u, v)$ is \emph{$(t, f)$ vertex (edge) protected} if there is no set $F$ of $f$ vertices (edges), with $u, v \notin F$,\footnote{The restriction $u, v \notin F$ is only meaningful when vertex protection is considered; in the case of edge protection there is no analogous requirement.} such that $\dist_{H \setminus F}(u, v) > t\cdot w(u, v)$.
\end{definition}
\begin{algorithm}[H]
\SetKwInOut{Input}{Input}
\SetKwInOut{Output}{Output}
\Input{An undirected weighted graph $G = (V, E, w)$ and positive integers $f, k$.}
Initialize $H \gets (V, \emptyset)$\;
\ForEach{$(u, v) \in E$ in order of ascending edge weight}{
\If{$(u, v)$ is not currently $(2k-1, f)$ vertex (edge) protected in $H$}{
add $(u, v)$ to $H$\;
}
}
\Return{$H$};

\Output{$H$, an $f$ VFT (EFT) $(2k-1)$-Spanner of $G$.}
\caption{\label{alg:spanners} Construction of $f$ VFT (EFT) $(2k-1)$-Spanners}
\end{algorithm}
%
%This algorithm mirrors the classic textbook greedy algorithm \cite{?} used to construct spanners in the non-faulty setting.

%Before proceeding to its analysis, we remark briefly on the efficiency of Algorithm \ref{alg:spanners}.
%A naive implementation of Algorithm \ref{alg:spanners} will take $n^{O(f)}$ time if one tests the protection of each incoming $(u, v)$ by brute-force iteration over all fault sets $F$ (and $n^{O(f)}$ is exponential, since the most interesting parameter range for this problem is $f = \text{poly}(n)$).
%However, this is unnecessary: one can instead test protection with an LP, which results in a polynomial runtime.
%\begin{definition}
%For an edge $(u, v) \in E$, the graph $H^{(u, v)}$ is defined as the subgraph of $H$ consisting of all edges considered before the edge $(u, v)$ in Algorithm \ref{alg:spanners} (not including $(u, v)$ itself).
%\end{definition}
%\begin{theorem}
%There is an implementation of Algorithm \ref{alg:spanners} that runs in polynomial time.
%\end{theorem}
%\begin{proof}
%It clearly suffices to show that, given a partially constructed spanner $H^{(u, v)}$, one can test whether or not $(u, v)$ is $(2k-1, f)$ vertex/edge protected in $H^{(u, v)}$.
%This is done as follows.
%Make $2k-1$ sequential copies of $H^{(u, v)}$, and for each edge $(x, y)$ in $H^{(u, v)}$, we place directed edges...
%\end{proof}
%It is possible -- perhaps even likely -- that 

It is essentially trivial to see that the output graph $H$ is indeed a $(2k-1)$ spanner of the input graph $G$:
\begin{theorem} \label{lem:alg-correct}
The graph $H$ returned by Algorithm \ref{alg:spanners} is an $f$ VFT (EFT) $(2k-1)$-spanner of the input graph $G=(V,E)$.
\end{theorem}
\begin{proof}
Let $F$ be any set of $f$ node (edge) faults.
Consider any $u,v\in V$ and consider a shortest path $\pi$ between $u$ and $v$ in $G\setminus F$.
Consider any edge $(x,y)$ on $\pi$. We have that either $(x,y)$ is in $H$ (and hence $H\setminus F$ since $(x,y)\in G\setminus F$), or the algorithm chose not to add it to $H$. If the latter event occurs, it must have been that $\dist_{H \setminus F}(x, y) \leq (2k-1)\cdot w(x, y)$ (by the protection definition, this holds for all fault sets). Thus, there is a path in $H\setminus F$ of weight at most $\sum_{(x,y)\in \pi}\dist_{H \setminus F}(x, y) \leq (2k-1) \dist_{G \setminus F}(u, v)$.
%
%By well-known folklore, it is necessary and sufficient to show that
%$$\dist_{H \setminus F}(u, v) \le (2k-1) \cdot w(u, v)$$
%for all edges $(u, v) \in G \setminus F$.
%For any such edge $(u, v)\in G \setminus F$, if $(u, v) \in H$ then we have $\dist_{H \setminus F}(u, v) = \dist_{G \setminus F}(u, v)$ (since $u, v$, or $(u, v)$ for edge faults, are not in $F$).
%Thus we only need to consider edges $(u, v) \notin H$.
%
%For any such edge, since we chose not to add $(u, v)$ to $H$ when it was considered, it follows that $(u, v)$ is $(2k-1, f)$ node (edge) protected in $H^{(u, v)}$.
%Thus, regardless of where the $F$ faults lie, by the protection definition, $\dist_{H \setminus F}(u, v) \leq (2k-1)\cdot w(u, v)$.
%Since all edges in $H^{(u, v)}$ have weight at most equal to the weight of $(u, v)$, the total weighted length of this non-faulted path is at most $(2k-1) \cdot w_{(u, v)}$.
The theorem follows.
\end{proof}

The vast majority of this paper is devoted to showing the upper bound $n^{1 + 1/k} f^{1 - 1/k} \cdot 2^{O(k)}$ on the density of $H$ as stated in Theorem~\ref{thm:upper-bounds}.

Our lower bounds are conditional on the standard Erd\"{o}s Girth Conjecture, which we recall:
\begin{conjecture} [Erd\"{o}s Girth Conjecture \cite{erdHos1964extremal}] \label{conj:erdos}
For any positive integer $k$, there exist infinite families of $n$-node graphs with $\Omega(n^{1 + 1/k})$ edges and girth $2k+2$.
\end{conjecture}

The paper is outlined as follows.
In the main body of this paper, we will highlight our matching upper and lower bounds for $k=2$ (i.e. $3$-spanners) and any $f$, for both edge and vertex faults.
The upper bound for $k=2$ is far simpler to prove than the upper bound for $k \ge 3$, and in some sense its supporting arguments form the base case for an inductive attack on larger $k$.
We will informally describe the extension of our argument to larger $k$, but due to space constraints, the formalities are deferred to the appendix.

%% file: lower-bound.tex
% !TEX root = newftspanners.tex

Details of these proofs can be found in Appendix~\ref{app:lbs}.  The proof of Theorem~\ref{thm:vertex-lbs} in its full generality is essentially the obvious generalization of the $k=2$ case.  We simply have to start with a graph $G$ with arises from the girth conjecture rather than from Lemma~\ref{lem:4-girth}.  We then use the same construction (making $f/2$ copies of each node, and turning each edge into a complete bipartite graph between the associated copies).  For any edge $\{(u,i), (v,j)\}$ in this graph, if the spanner does not include this edge and if we fail all other copies of $u$ and $v$, the shortest path from $(u,i)$ to $(v,j)$ in the spanner must correspond to a walk from $u$ to $v$ in $G$ which does not use the edge $\{u,v\}$.  But by construction $G$ has large girth, so all such walks are too long.  Hence $\{(u,i), (v,j)\}$ must be in the spanner.

In the case of edge faults (Theorem~\ref{thm:edge-lbs}), when $k > 2$ the same construction unfortunately no longer works.  To see why, consider as before some edge $\{(u,i), (v,j)\}$ which is in the graph but not the spanner.  If we use the same fault set $F$ as in the $k=2$ case, there will still be a path in $G \setminus F$ of length $5$ from $(u,i)$ to $(v,j)$ of the form $(u,i), (x, a), (u,b), (v, c), (y, d), (v, j)$.  Here $x$ is any neighbor of $u$ in $G$ and $y$ is any neighbor of $v$ and $a,b,c,d$ are arbitrary indices.  This will be the case even if the original graph $G$ has large girth.  

To get around this problem, we make even fewer copies of each vertex.  In particular, we will make only $t = \lfloor \sqrt{f} \rfloor$ copies of each vertex.  Now, instead of our fault set being all edges from $(u,i)$ to copies of $v$ and all edges from $(v,j)$ to copies of $u$, we can simply remove all edges between copies of $u$ and copies of $v$ except for $\{(u,i), (v,j)\}$.  Now the same logic as before implies that this edge must be in the spanner, but when we analyze the size of the graph we get only $\Omega\left(f^{1/2 - 1/(2k)} |V_{G'}|^{1+1/k}\right)$.

%% file: app-lower.tex
% !TEX root = newftspanners.tex
 \section{Lower Bounds for Larger $k$} \label{app:lbs}
 
 We now prove our lower bounds in full generality.   
 
 \subsection{Proof of Theorem~\ref{thm:vertex-lbs}}
 %\gre{In this proof, can't we have $3$-paths of the form $\{(u, i), (v, j' \ne j), (w, k), (v, j)\}$?  I think the claim that $x^a \ne x^b$ below breaks in this case because we have $(x^{a-1}, i^{a-1}) = (u, i)$ and $(x^b, i^b) = (v, j)$ and we removed that edge.
 %In any case, I think the proof will still go through if we use clouds of size $\lfloor f/2 \rfloor + 1$ and then fault $F = \{(u, \ell) \ \mid \ \ell \ne i\} \cup \{(v, \ell') \ \mid \ \ell' \ne j\}$.}
 
Let $G$ be a graph from the girth conjecture (Conjecture~\ref{conj:erdos}), i.e., a graph with girth at least $2k+2$ and $\Omega(n^{1+1/k})$ edges.  We construct a new graph $G'$ as follows.  Let $t = \lceil f/2 \rceil$.  We set $V_{G'} = V(G) \times [t]$, and let $E_{G'} = \{ \{(u,i), (v,j)\} : \{u,v\} \in E(G) \land i,j \in [t]\}$.  Let $G' = (V_{G'}, E_{G'})$.  Intuitively, we can think of $G'$ as being obtained by replacing each vertex of $G$ by a set of $t$ copies of the vertex, and each edge of $G$ is replaced by a complete bipartite graph between the two sets of copies.  We will prove Theorem~\ref{thm:vertex-lbs} by proving that the only $f$ VFT $(2k-1)$-spanner of $G'$ is itself, and that it has the required number of edges.

We first claim that the only $f$ VFT $(2k-1)$-spanner of $G'$ is $G'$ itself.  To see this, suppose that $H$ is a subgraph of $G'$ which does not contain some edge $\{(u,i), (v,j)\} \in E_{G'}$.  Let $F = \{(u,\ell) : \ell \neq i\} \cup \{(v, \ell) : \ell \neq j\}$, i.e., we let the fault set be all copies of $u$ except for $(u,i)$ and all copies of $v$ except $(v,j)$.  Note that $|F| \leq f$.  Now consider the shortest path from $(u,i)$ to $(v,j)$ in $H \setminus F$.  Let this path be $(u,i) = (x^0, i^0), (x^1, i^1), (x^2, i^2), \dots, (x^p, i^p) = (v,j)$.  Note that for any $0 \leq a \leq p-1$, it cannot be the case that $x^a = u$ and $x^{a+1} = v$, since no such edges exist in $H \setminus F$.  Thus $\langle u = x^0, x^1, \dots, x^p = v$ is a (possibly non-simple) path from $u$ to $v$ in $G$ which does not use the edge $\{u,v\}$.  By adding $\{u,v\}$ to this path, we get a cycle of length at most $p+1$.  Since $G$ has girth at least $2k+2$, this implies that $p \geq 2k+1$, and thus that in $H \setminus F$ the distance between $(u,i)$ and $(v,j)$ is at least $2k+1$, while in $G' \setminus F$ they are at distance $1$.  Thus $H$ is not an $f$ VFT $(2k-1)$-spanner of $G'$, and hence the only $f$ VFT $(2k-1)$-spanner of $G'$ is $G'$ itself.

%
%Since this is the shortest path and the only remaining copy of $u$ is $(u,i)$, we know that $x^a \neq u$ for all $a > 0$.  Similarly, since it is a shortest path and each edge of $G$ has been replaced by a complete bipartite graph, we know that if $0 < a < b$ then $x^a \neq x^b$ (or else there would be a shorter path which would go directly from $(x^{a-1}, i^{a-1})$ to $(x^b, i^b)$).  Hence $u = x^0, x^1, \dots, x^p = v$ must be a simple $u-v$ path in $G$, and so when combined with the edge $\{u,v\}$ forms a simple cycle in $G$ of length $p+1$.  But $G$ has girth at least $2k+2$, and hence $p \geq 2k+1$.  Thus $H$ is not a $(2k-1)$-spanner of $G' \setminus F$, and thus is not a $f$ VFT $(2k-1)$-spanner of $G'$.  

So $G'$ is the only $f$ VFT $(2k-1)$-spanner of itself, and it remains only to analyze its size.  Clearly $|V_{G'}| = t |V(G)|$, and by assumption $|E(G)| \geq \Omega(|V(G)|^{1+1/k})$.  Thus
\begin{align*}
|E_{G'}| &= |E(G)| t^2 = \Omega(t^2 |V(G)|^{1+1/k}) = \Omega\left(t^2 \cdot \left(\frac{|V_{G'}|}{t}\right)^{1+1/k}\right) \\
&= \Omega\left(f^{1 - 1/k} |V_{G'}|^{1+1/k}\right). 
\end{align*}

\subsection{Proof of Theorem~\ref{thm:edge-lbs}}

Our proof is very similar to the proof of Theorem~\ref{thm:vertex-lbs}.  The case of $k=2$ was proved in Section~\ref{sec:lower-2}, so we only prove the theorem for $k \geq 3$, where we are forced to use a slightly different construction which will give only the weaker bound claimed in Theorem~\ref{thm:edge-lbs}.

Let $t = \lfloor \sqrt{f} \rfloor$.  We set $V_{G'} = V(G) \times [t]$, and as before set $E_{G'} = \{ \{(u,i), (v,j)\} : \{u,v\} \in E(G) \land i,j \in [t]\}$ (where $G$ is a graph from the girth conjecture).  Let $H$ be a subgraph of $G'$ missing some edge $e = \{(u,i), (v,j)\}$.  As in the other lower bounds, we will prove that $H$ cannot be an $f$ EFT $(2k-1)$ spanner of $G'$, and thus the only $f$ EFT $(2k-1)$-spanner of $G'$ is $G'$ itself.  

Consider the fault set $F = \{ \{ (u, i'), (v, j') \} : i' \neq i \lor j' \neq j\}$.  In other words, we remove all edges between copies of $u$ and copies of $v$ \emph{except} for the edge $e$ (which is therefore in $G' \setminus F$ but not in $H \setminus F$).  Now  let $P = \langle (u,i) = (x^0, i^0), (x^1, i^1), (x^2, i^2), \dots, (x^p, i^p) = (v,j) \rangle$ be the shortest path from $(u,i)$ to $(v,j)$ in $H \setminus F$.  Note that for any $0 \leq a \leq p-1$, it cannot be the case that $x^a = u$ and $x^{a+1} = v$, since no such edges exist in $H \setminus F$.  Thus $\langle u = x^0, x^1, \dots, x^p = v \rangle$ is a (possibly non-simple) path from $u$ to $v$ in $G$ which does not use the edge $\{u,v\}$.  By adding $\{u,v\}$ to this path, we get a cycle of length at most $p+1$.  Since $G$ has girth at least $2k+2$, this implies that $p \geq 2k+1$, and thus that in $H \setminus F$ the distance between $(u,i)$ and $(v,j)$ is at least $2k+1$, while in $G' \setminus F$ they are at distance $1$.  Thus $H$ is not an $f$ EFT $(2k-1)$-spanner of $G'$, and hence the only $f$ EFT $(2k-1)$-spanner of $G'$ is $G'$ itself.

It remains only to analyze the size of $G'$.  By construction, $|V_{G'}| = t |V(G)|$, and by assumption $|E(G)| \geq \Omega(|V(G)|^{1+1/k})$.  Thus
\begin{align*}
|E_{G'}| &= |E(G)| t^2 = \Omega(|V(G)|^{1+1/k} t^2) = \Omega\left(t^2 \cdot \left(\frac{|V_{G'}|}{t}\right)^{1+1/k}\right) \\
&= \Omega\left(t^{1 - 1/k} |V_{G'}|^{1+1/k}\right) = \Omega\left(f^{1/2 - 1/(2k)} |V_{G'}|^{1+1/k}\right). 
\end{align*}

\subsection{Strong Incompressibility}

We now show how to generalize the above proofs to obtain strong incompressibility (Theorems \ref{thm:vertex-incomp} and \ref{thm:edge-incomp}).
This type of argument was first used by Matou{\v{s}}ek \cite{Mat96}.

The proof of Theorem \ref{thm:vertex-incomp} (incompressibility under vertex faults) goes as follows.
Let $G^* = (V, E^*)$ be the graph described in the proof of Theorem \ref{thm:vertex-lbs}.
There are $2^{|E^*|}$ possible subgraphs of $G^*$ on the same vertex set.
Consider two distinct such subgraphs $G_1 = (V, E_1) \ne G_2 = (V, E_2)$, let $\dee_1$ be the data structure created by processing $G_1$, let $\dee_2$ be the data structure created by processing $G_2$, and let $((u, i), (v, j)) \in E_1 \setminus E_2$.
As before, let
$$F := \left\{(u, \ell) : \ell \ne i \right\} \cup \left\{(v, \ell) : \ell \ne j\right\}.$$
By the previous argument, we then have
$$\dist_{G_1 \setminus F}((u, i), (v, j)) = 1 \quad \text{ and } \quad \dist_{G_2 \setminus F}((u, i), (v, j)) \ge 2k+1$$
and so
$$\widehat{\dist_{G_1 \setminus F}}((u, i), (v, j)) \le 2k-1 \quad \text{ and } \quad \widehat{\dist_{G_2 \setminus F}}((u, i), (v, j)) \ge 2k+1.$$
Thus $\dee_1, \dee_2$ produce different answers to the query $((u,i),(v,j),F)$, so they must have different representations.
Since this holds for any two subgraphs $G_1, G_2 \subseteq G^*$, we have a family of $2^{|E^*|}$ graphs that all have different representations, and so by the pigeonhole principle the data structure $\dee_G$ for one of these subgraphs $G \subseteq G^*$ occupies at least $|E^*| = \Omega\left(f^{1 - 1/k} n^{1 + 1/k}\right)$ bits of space, thus completing the proof.

The proof of Theorem \ref{thm:edge-incomp} is identical, taking $F$ to be a set of edge faults as described in the above proof of Theorem \ref{thm:edge-lbs}.